\documentclass[a4paper,USenglish,cleveref,autoref]{lipics-v2019}

\usepackage{dsfont}
\usepackage{graphicx}
\usepackage[utf8]{inputenc}
\usepackage[T1]{fontenc}
\usepackage{paralist}
\usepackage{subcaption}

\usepackage{color}
\usepackage[dvipsnames,table]{xcolor}
\usepackage{amssymb,amsfonts,amsmath,amsthm,thm-restate}
\definecolor {sepia} {rgb} {0.75,0.30,0.15}
\usepackage{cleveref}
\usepackage{lineno}

\title{Simultaneous Representation of Proper and Unit Interval Graphs}
\titlerunning{Simultaneous Proper and Unit Interval Graphs}

\usepackage{xspace}
\usepackage[ruled, vlined, linesnumbered]{algorithm2e}

\newcommand{\co}{\ensuremath{\zeta}\xspace}
\newcommand{\oo}{\ensuremath{\sigma}\xspace}

\newcommand{\conc}{\ensuremath{\circ}\xspace}

\newcommand{\consistent}{{\text{\sc Consistent}}}

\makeatletter
\newcommand{\cOrd}[2]{\begin{math} \newcommand{\Symbol}{#1} #2 \@BeginOrdering}
  \newcommand\@BeginOrdering{\@ifnextchar\eOrd{\@EndOrdering}{\@ContinueOrdering}}
  \newcommand\@ContinueOrdering[1]{\Symbol #1\@BeginOrdering}
  \newcommand\@EndOrdering[1]{\end{math}\xspace}
\makeatother

\author{Ignaz Rutter}{Faculty of Computer Science and Mathematics, University of
Passau, Germany}{rutter@fim.uni-passau.de}{https://orcid.org/0000-0002-3794-4406}{}

\author{Darren Strash}{Department of Computer Science, Hamilton College,
USA}{dstrash@hamilton.edu}{https://orcid.org/0000-0001-7095-8749}{}

\author{Peter Stumpf}{Faculty of Computer Science and Mathematics, University of
Passau, Germany}{stumpf@fim.uni-passau.de}{https://orcid.org/0000-0003-0531-9769}{}

\author{Michael Vollmer}{Department of Informatics, Karlsruhe Institute of
Technology (KIT), Germany}{michael.vollmer@kit.edu}{}{}

\funding{This work was partially supported by grant RU 1903/3-1 of the German Research Foundation (DFG) and by the DFG Research Training Group 2153: ``Energy Status Data -- Informatics Methods for its Collection, Analysis and Exploitation''.}

\authorrunning{I.\ Rutter, D.\ Strash, P.\ Stumpf, M.\ Vollmer}

\Copyright{Ignaz Rutter, Darren Strash, Peter Stumpf, Michael Vollmer}



\category{}
\relatedversion{}
\supplement{}
\nolinenumbers

\begin{document}

\maketitle

\begin{abstract} 
  In a confluence of combinatorics and geometry, simultaneous representations provide a way to realize combinatorial objects that share common structure. 
  A standard case in the study of simultaneous representations is the \emph{sunflower case} where all objects share the same common structure. While the recognition problem for general simultaneous interval graphs---the simultaneous version of arguably one of the most well-studied graph classes---is NP-complete, the complexity of the sunflower case for three or more simultaneous interval graphs is currently open. 
  In this work we settle this question for \emph{proper} interval graphs. We give an algorithm to recognize simultaneous proper interval graphs in linear time in the sunflower case where we allow any number of simultaneous graphs. 
  Simultaneous \emph{unit} interval graphs are much more `rigid' and therefore have less freedom in their representation. We show they can be recognized in time $\mathcal{O}(|V|\cdot|E|)$ for any number of simultaneous graphs in the sunflower case where $G=(V,E)$ is the union of the simultaneous graphs. 
 We further show that both recognition problems are in general NP-complete if the number of simultaneous graphs is not fixed. The restriction to the sunflower case is in this sense necessary.
\end{abstract}

\section{Introduction}
\label{ch:introduction}

Given a family of sets $\mathcal R$, the corresponding \emph{intersection graph} $G$ has a vertex for each set and two vertices are adjacent if and only if their sets have a non-empty intersection. If all sets are intervals on the real line, then $\mathcal R$ is an \emph{interval representation} of $G$ and $G$ is an \emph{interval graph}; see Figure~\ref{intro:fig:graphs}.

\begin{figure}[tb]
	\centering
	\includegraphics[scale=1]{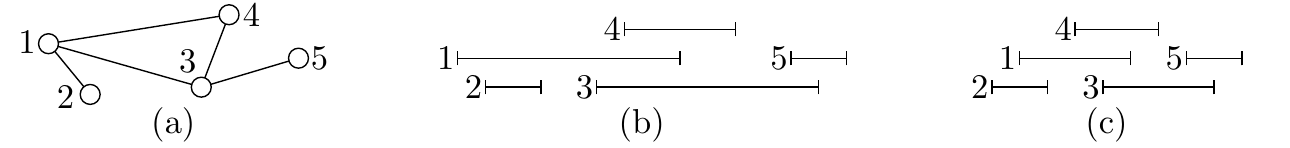}
	\caption{ (a) A graph, with (b) an interval representation and (c) proper interval representation.}
	\label{intro:fig:graphs}
\end{figure}

In the context of intersection graph classes, much work has been devoted to efficiently computing a \emph{representation}, which is a collection of sets or geometric objects having an intersection graph that is isomorphic to a given graph. For many well-known graph classes, such as interval graphs and chordal graphs, this is a straightforward task~\cite{Golumbic:2004:AGT:984029,spinrad-2003}. However, often it is desirable to consistently represent \emph{multiple} graphs that have subgraphs in common. 
This is true, for instance, in realizing schedules with shared events, embedding circuit graphs of adjacent layers on a computer chip, and visualizing the temporal relationship of graphs that share a common subgraph~\cite{jampani-jgaa-2012}. Likewise, in genome reconstruction, we can ask if a sequence of DNA can be reconstructed from strands that have sequences in common~\cite{goldberg1995four}. 

\emph{Simultaneous representations} capture this in a very natural way. Given \emph{simultaneous graphs} $G_1,G_2,\ldots,G_k$ where each pair of graphs $G_i, G_j$ share some common subgraph, a \emph{simultaneous representation} asks for a fixed representation of each vertex that gives a valid representation of each $G_i$.  This notion is closely related to \emph{partial representation extension}, which asks if a given (fixed) representation of a subgraph can be extended to a representation of the full graph. Partial representation extension has been extensively studied for graph classes such as interval graphs \cite{Klavik2017}, circle graphs \cite{Chaplick2013}, as well as proper and unit interval graphs \cite{Klavik2017}.
For interval graphs, Bläsius and Rutter \cite{Blasius:2015:SPA:2846106.2738054} have even shown that the partial interval representation problem can be reduced to a simultaneous interval representation problem on two graphs in linear time. 

Simultaneous representations were first studied in the context of embedding graphs~\cite{DBLP:journals/corr/abs-1204-5853,brass2007simultaneous}, where the goal is to embed each simultaneous graph without edge crossings while shared subgraphs have the same induced embedding. Unsurprisingly, many variants are NP-complete~\cite{Gassner2006, Schaefer2013, Angelini2014,Estrella-Balderrama2008}. The notion of simultaneous representation of general intersection graph classes was introduced by Jampani and Lubiw~\cite{jampani-jgaa-2012}, who showed that it is possible to recognize simultaneous chordal graphs with two graphs in polynomial time, and further gave a polynomial time algorithm to recognize simultaneous comparability graphs and permutation graphs with two or more graphs that share the same subgraph (the \emph{sunflower case}). They further showed that recognizing three or more simultaneous chordal graphs is NP-complete.

Golumbic et al.~\cite{golumbic1995graph} introduced the \emph{graph sandwich problem} for a graph class $\Pi$. Given a vertex set $V$ and edge sets $E_1\subseteq E_2\subseteq \binom{V}{2}$ it asks whether there is an edge set $E_1\subseteq E\subseteq E_2$ such that the \emph{sandwich graph} $G=(V,E)$ is in $\Pi$. Jampani and Lubiw showed that if $\Pi$ is an intersection graph class, then recognizing $k$ simultaneous graphs in $\Pi$ in the sunflower case is a special case of the graph sandwich problem where $(V,E_2\setminus E_1)$ is a $k$-partite graph~\cite{jampani-jgaa-2012}.

We consider simultaneous \emph{proper} and \emph{unit} interval graphs. An interval graph is proper if in an interval representation no interval properly contains another one (see Figure~\ref{intro:fig:graphs}), and it is unit if all intervals have length~one.
Interestingly, while proper and unit interval graphs are the same graph class as shown by Roberts \cite{Roberts:Indifference}, simultaneous unit interval graphs differ from simultaneous proper interval graphs; see Figure~\ref{fig:SPneqSU}.  Unit interval graphs are intersection graphs and therefore the graph sandwich paradigm described by Jampani and Lubiw applies. Proper interval graphs are not since in a simultaneous representation intervals of distinct graphs may contain each other which means that the intersection graph of all intervals in the simultaneous representation is not proper.

Sunflower (\emph{unit}) interval graphs are a generalization of \emph{probe} (proper) interval graphs, where each sunflower graph has only one non-shared vertex. Both variants of probe graphs can be recognized in linear time~\cite{mcconnell2009linear,nussbaum2014recognition}.

\begin{figure}[tb]
	\centering
	\includegraphics{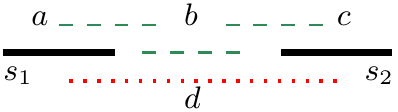}
	\caption{A simultaneous proper interval representation of a sunflower graph $\mathcal{G}$ consisting of two paths $G_1=(s_1,a,b,c,s_2)$ (dashed) and $G_2=(s_1,d,s_2)$ (dotted) with shared start and end $s_1,s_2$ (bold). They have no simultaneous unit interval representation: The intervals $a$ and $c$ enforce that $b$ lies between $s_1$ and $s_2$. Interval $d$ therefore includes $b$ in every simultaneous proper interval representation. In particular, not both can have size one.}
	\label{fig:SPneqSU}
\end{figure}

Simultaneous interval graphs were first studied by Jampani and Lubiw~\cite{Jampani2010} who gave a $\mathcal{O}(n^2\lg n)$-time recognition algorithm for the special case of two simultaneous graphs. Bl\"asius and Rutter~\cite{Blasius:2015:SPA:2846106.2738054} later showed how to recognize two simultaneous interval graphs in linear time. Bok and Jedli{\v{c}}kov{\'a} showed that the recognition of an arbitrary number of simultaneous interval graphs is in general NP-complete~\cite{bok2018note}. However, the complexity for the sunflower case with more than two simultaneous graphs is still open.

\subparagraph{Our Results.}
We settle these problems with $k$ not fixed for simultaneous \emph{proper} and \emph{unit} interval graphs -- those
graphs with an interval representation where no interval properly contains
another and where all intervals have unit length, respectively~\cite{Deng:1996:LRA:586666.586674, JGT:JGT3190060307, deFigueiredo1995179, Heggernes2015252}. 
For the sunflower case, we provide efficient recognition algorithms. The running time for proper interval graphs
is linear, while for the unit case it is $\mathcal{O}(|V|\cdot |E|)$ where $G = (V,E)$ is the union of the sunflower graphs.
In Appendix~\ref{ch:hardness} we prove NP-completeness for the  non-sunflower case. The reductions are similar to the simultaneous independent work of Bok and Jedli{\v{c}}kov{\'a} for simultaneous interval graphs~\cite{bok2018note}.

\subparagraph{Organization.}We begin by introducing basic notation and existing tools throughout Section~\ref{ch:preliminaries}. In Section~\ref{ch:SConnected} we give a characterization of simultaneous proper interval graphs, from which we develop an efficient recognition algorithm. In Section~\ref{sec:unit} we characterize simultaneous proper interval graphs that can be simultaneous unit interval graphs, and then exploit this property to efficiently search for a representation among simultaneous proper interval graph representations. 
Proofs of lemmas and theorems marked with $\star$ are provided in the appendix.

\section{Preliminaries}
\label{ch:preliminaries}
In this section we give basic notation, definitions and
characterizations.  Section~\ref{prel:sec:graphs} collects basic
concepts on graph theory, orderings, and PQ-trees.
Section~\ref{prel:sec:interval} introduces (proper) interval graphs
and presents relations between the representations of such graphs and
their induced subgraphs.  Finally, Section \ref{prel:sec:simul}
introduces the definition and notation of simultaneous graphs.

\subsection{Graphs, Orderings, and PQ-trees}
\label{prel:sec:graphs}

Unless mentioned explicitly, all graphs in this paper are undirected.
For a graph $G=(V,E)$ we denote its size $|G| := |V| + |E|$.

Let $\sigma$ be a binary relation. Then we write $a_1\le_\sigma a_2$ for $(a_1,a_2)\in\sigma$, and we write $a_1<_\sigma a_2$ if $a_1\leq_\sigma a_2$ and $a_1\ne a_2$.
We omit the subscript and simply use $<$ and $\le$ if the ordering it refers to is clear from
the context.  We denote the \emph{reversal} of a linear order $\oo$
by $\oo^r$, and we use $\circ$ to concatenate linear orders of
disjoint sets.

A \emph{PQ-tree} is a data structure for representing sets of linear
orderings of a ground set $X$.  Namely, given a set
$\mathcal C \subseteq 2^X$, a \emph{PQ-tree on $X$ for $\mathcal C$}
is a tree data structure $T$ that represents the set $\consistent(T)$
containing exactly the linear orders of $X$ in which the elements
of each set $C \in \mathcal C$ are consecutive.  The PQ-tree $T$ can
be computed in time
$O(|X| + \sum_{C \in \mathcal C} |C|)$~\cite{BOOTH1976335}.  Given a
PQ-tree $T$ on the set $X$ and a subset $X' \subseteq X$, there exists
a PQ-tree $T'$, called the \emph{projection} of $T$ to $X'$, that
represents exactly the linear orders of $X'$ that are restrictions
of orderings in $\consistent(T)$.  For any two PQ-trees $T_1$ and
$T_2$ on the set $X$, there exists a PQ-tree $T$ with
$\consistent(T) = \consistent(T_1) \cap \consistent(T_2)$, called the
\emph{intersection} of $T_1$ and $T_2$.  Both the projection and the
intersection can be computed in $O(|X|)$ time~\cite{BoothPhd}.

\subsection{Interval Graphs, Proper Interval Graphs, and Their Subgraphs}
\label{prel:sec:interval}

An \emph{interval representation} $R = \{I_v \mid v \in V\}$ of a
graph $G=(V,E)$ associates with each vertex $v \in V$ an interval
$I_v = [x,y]$ of real numbers such that for each pair of vertices
$u,v \in V$ we have $I_u \cap I_v \ne \emptyset$ if and only if
$\{u,v\} \in E$, i.e., the intervals intersect if and only if the
corresponding vertices are adjacent.  An interval representation $R$
is \emph{proper} if no interval properly contains another one, and it
is \emph{unit} if all intervals have length~1.  A graph is an
\emph{interval graph} if and only if it admits an interval
representation, and it is
a \emph{proper (unit) interval graph} if and only if it admits a
proper (unit) interval representation.  It is well-known that proper and unit interval graphs are the same graph class.

\begin{proposition}[\cite{Roberts:Indifference}]
  \label{the:unitproperEqu}
  A graph is a unit interval graph if and only if it is a proper
  interval graph.
\end{proposition} 

However, this does not hold in the simultaneous case where every simultaneous unit interval representation is clearly a simultaneous proper interval representation of the same graph, but not every simultaneous proper interval representation implies a simultaneous unit interval representation; see Figure~\ref{fig:SPneqSU}.

We use the well-known characterization of proper interval graphs using \emph{straight enumerations}~\cite{Deng:1996:LRA:586666.586674}.
Two adjacent vertices $u,v\in V$ are \emph{indistinguishable} if we have $N[u]=N[v]$ where $N[u]=\{v\colon uv\in E(H)\}\cup \{u\}$ is the closed neighborhood.
Being indistinguishable is an equivalence relation and we call the equivalence classes \emph{blocks} of $G$.  
We denote the block of $G$ that contains vertex $u$ by $B(u,G)$. Note that for a subgraph $G'\subseteq G$ the block $B(u,G')$ may contain vertices in $V(G')\setminus B(u,G)$ that have the same neighborhood as $u$ in $G'$ but different neighbors in $G$. 
Two blocks $B$, $B'$ are \emph{adjacent} if and only if $uv\in E$ for (any) $u\in B$ and $v\in B'$.
A linear order $\sigma$ of the blocks of $G$ is a \emph{straight enumeration} of $G$ if for every block, the block
and its adjacent blocks are consecutive in $\sigma$.  A proper interval
representation $R$ defines a straight enumeration $\sigma(R)$ by ordering
the intervals by their starting points and grouping together the
blocks.  Conversely, for each
straight enumeration $\sigma$, there exists a corresponding
representation $R$ with
$\sigma = \sigma(R)$~\cite{Deng:1996:LRA:586666.586674}.  A \emph{fine enumeration} of a graph $H$ is a linear order $\eta$ of $V(H)$ such that for $u\in V(H)$ the closed neighborhood $N_H[u]$ is consecutive in $\eta$.

\begin{proposition}[\cite{roberts1968representations,Deng:1996:LRA:586666.586674,Hell:2002:FDA:586843.586923}]
\label{proper:theorem:straightEnum}
\label{proper:theorem:uniqueEnumeration}
\label{proper:fine-enumeration}
\begin{inparaenum}[(i)]
\item  A graph is a proper interval graph if and only if it has a fine enumeration.
\item A graph is a proper interval graph if and only if it admits a
  straight enumeration.
\item A straight enumeration of a connected proper interval graph is
  unique up to reversal.
\end{inparaenum}
\end{proposition}

\subsection{Simultaneous Graphs}
\label{prel:sec:simul}

A \emph{simultaneous graph} is a tuple $\mathcal G = (G_1, \dots, G_k)$ of
graphs $G_i$ that may each share vertices and edges. Note that this definition
differs from the one we gave in the introduction. This way the input for the
simultaneous representation problem is a single entity.  The size
$|\mathcal{G}|$ of a simultaneous graph is $\sum_{i = 1}^k |G_i|$. We call
$\mathcal G$ \emph{connected}, if $\bigcup_{i=1}^k G_i$ is connected. A
\emph{simultaneous (proper/unit) interval representation} $\mathcal{R} = (R_1,
\dots, R_k)$ of $\mathcal{G}$ is a tuple of representations such that $R_i \in
\mathcal{R}$ is a (proper/unit) interval representation of graph $G_i$ and the
intervals representing shared vertices are identical in each representation.  A
simultaneous graph is a \emph{simultaneous (proper/unit) interval graph} if it
admits a simultaneous (proper/unit) interval representation.

An important special case is that of \emph{sunflower graphs}.  The
simultaneous graph $\mathcal G$ is a \emph{sunflower} graph if each
pair of graphs $G_i,G_j$ with $i \ne j$ shares exactly the same
subgraph $S$, which we then call the \emph{shared graph}.
Note that, for $\mathcal G$ to be a simultaneous interval graph, it is
a necessary condition that $G_i \cap G_j$ is an induced subgraph of
$G_i$ and $G_j$ for $i,j=1,\dots,k$.  In particular, in the sunflower
case the shared graph $S$ must be an induced subgraph of each $G_i$.
The following lemma allows us to restrict ourselves to instances whose
union graph $\bigcup_{\mathcal G} = \bigcup_{i=1}^k G_i$ is connected.

\begin{restatable}[$\star$]{lemma}{unionconnected}
  \label{lem:union-connected}
  Let $\mathcal G = (G_1,\dots,G_k)$ be a simultaneous graph and let
  $C_1,\dots,C_l$ be the connected components of~$\bigcup_{\mathcal G}$.
  Then $\mathcal G$ is a simultaneous (proper) interval graph if and
  only if each of the graphs $\mathcal G_i = (G_1 \cap C_i,\dots,G_k
  \cap C_i)$, $i=1,\dots,l$ is a simultaneous (proper/unit) interval
  graph.
\end{restatable}

\section{Sunflower Proper Interval Graphs}
\label{ch:SConnected}

In this section, we deal with simultaneous proper interval
representations of sunflower graphs.  We first present a combinatorial
characterization of the simultaneous graphs that admit such a
representation.  Afterwards, we present a simple linear-time
recognition algorithm.  Finally, we derive a combinatorial description of
all the combinatorially different simultaneous proper interval
representations of a connected simultaneous graph, which is a
prerequisite for the unit case.

\subsection{Characterization}
\label{simul:sec:rest}

Let $G = (V, E)$ be a proper interval graph with straight enumeration \oo and let $V_S \subseteq V$ be a subset of vertices.
We call \oo \emph{compatible} with a linear order \co of $V_S$ if, we have for $u,v \in V_S$ that \cOrd{\leq_\co}{u}{v}\eOrd implies \cOrd{\leq_\oo}{B(u, G)}{B(v, G)}\eOrd.

\begin{restatable}{lemma}{enumCombinationLemma}
\label{lem:EnumCombination}
Let $\mathcal{G} = (G_1, \dots, G_k)$ be a sunflower graph with shared
graph $S = (V_S, E_S)$.  Then $\mathcal G$ admits a simultaneous proper
interval representation $\mathcal R$ if and only if
there exists a linear order $\co$ of $V_S$ and straight enumerations $\oo_i$
for each $G_i$ that are compatible with $\co$.
\end{restatable}
{\renewcommand{\proofname}{Proof Sketch} \begin{proof}
For a given representation $\mathcal{R}$ the straight enumerations $\oo_i = \oo(R_i)$ and linear order \co of $V_S$ given by their left endpoints in $\mathcal{R}$ clearly satisfy the lemma.
Conversely we build a linear order of interval endpoints from each $\oo_i$ that equals a proper interval representation.
As each $\oo_i$ is compatible with \co, all endpoint orderings allow the same ordering for vertices in $S$, thus permitting one global ordering of all endpoints.
Drawing the intervals according to this ordering then yields a simultaneous representation $\mathcal{R}$ since it extends each individual ordering.
\end{proof}}

Let $\mathcal{G}=(G_1,\dots,G_k)$ be a sunflower graph with shared graph $S = (V_S, E_S)$ and for each $G_i \in \mathcal{G}$ let $\oo_i$ be a straight enumeration of $G_i$.
We call the tuple $(\oo_1, \dots, \oo_k)$ a \emph{simultaneous enumeration} if
for any $i,j \in \{1, \dots, k\}$ and $u,v \in V_S$ we have \cOrd{<_{\oo_i}}{B(u,G_i)}{B(v,G_i)}\eOrd $\Rightarrow$ \cOrd{\leq_{\oo_j}}{B(u,G_j)}{B(v,G_j)}\eOrd.
That is, the blocks containing vertices of the shared graph are not ordered differently in any straight enumeration.

\begin{restatable}[$\star$]{theorem}{simulEnumerationstheorem}
  \label{the:simulEnums}
  Let $\mathcal G = (G_1,\dots,G_k)$ be a sunflower graph.
  There exists a simultaneous proper interval representation $\mathcal{R} = (R_1, \dots, R_k)$ of $\mathcal{G}$ if and only if there is a simultaneous enumeration $(\oo_1, \dots, \oo_k)$ of $\mathcal{G}$.
  If $(\oo_1, \dots, \oo_k)$ exists, there also exists $\mathcal{R}$ with $\oo(R_i) = \oo_i$ for each $R_i \in \mathcal{R}$.
\end{restatable}

\subsection{A Simple Recognition Algorithm}
\label{sec:simple-algorithm}

In this section we develop a very simple recognition algorithm for
sunflower graphs that admits a simultaneous proper interval representation based on Theorem~\ref{the:simulEnums}.

Let~$\mathcal G = (G_1,\dots,G_k)$ be a sunflower graph with shared
graph $S=(V_S,E_S)$.  By Proposition~\ref{proper:theorem:straightEnum}, for each
graph $G_i$, there exists a PQ-tree $T_i'$ that describes exactly the
fine enumerations of $G_i$.  We denote by~$T_i = T_i'|S$ the
projection of $T_i$ to the vertices in $S$.  The tree $T_i$ thus
describes all proper interval representations of~$S$ that can be
extended to a proper interval representation of~$G_i$.  Let $T$ denote
the intersection of $T_1,\dots,T_k$.  By definition, $T$ represents
all proper interval representations of $S$ that can be extended to a
proper interval representation of each graph $G_i$.  Thus,
$\mathcal G$ admits a simultaneous proper interval representation if
and only if~$T$ is not the null-tree.

If~$T$ is not the null-tree, we can obtain a simultaneous
representation by choosing any ordering $O \in \consistent(T)$ and
constructing a simultaneous representation~$\mathcal S$ of~$S$.  Using
the algorithm of~Klav{\'i}k et al.~\cite{Klavik2017}, we can then
independently extend $\mathcal S$ to representations~$R_i$ of
$G_i$.  Since the trees~$T_i$ can be computed in time linear in the
size of the graph~$G_i$, and the intersection of two trees takes
linear time, the testing algorithm takes time linear in the total size
of the~$G_i$.  The representation extension of Klav{\'i}k et al.~\cite{Klavik2017} runs
in linear time.  We therefore have the following theorem.

\begin{theorem}
  \label{the:spirRecognition}
  Given a sunflower graph~$\mathcal G=(G_1,\dots,G_k)$, it can be
  tested in linear time whether~$\mathcal G$ admits a simultaneous
  proper interval representation.
\end{theorem}

\subsection{Combinatorial Description of Simultaneous Representations}

Let $\mathcal{G}$ be a sunflower proper interval graph with shared
graph $S$ and simultaneous representation $\mathcal{R}$.  Then, each
representation $R \in \mathcal{R}$ uses the same intervals for
vertices of $S$ and implies the same straight enumeration
$\oo_S(\mathcal{R}) = \oo_S(R) = \oo(\{I_v \in R : v \in V(S)\})$.

\begin{lemma}
\label{simul:lemma:fixSubgraphOrdering}
Let $\mathcal{G}$ be a connected sunflower proper interval graph with shared
graph $S$.  Across all simultaneous proper interval representations
$\mathcal{R'}$ of $\mathcal{G}$, the straight enumeration
$\oo_S(\mathcal{R})$ of $S$ is unique up to reversal.
\end{lemma}
\begin{proof}
  Let $\mathcal{R}$ be a simultaneous representation of $\mathcal{G}$
  and $\oo_S(\mathcal{R})$ the straight enumeration of $S$ induced by
  $\mathcal{R}$.  Since $\mathcal{G}$ is connected, for any two blocks $B_i$ and $B_{i+1}$ of $S$
  consecutive in $\oo_S(\mathcal{R})$, there exists a graph
  $G \in \mathcal{G}$ such that $B_i$ and $B_{i+1}$ are in the same
  connected component of $G$.  Since $S$ is an induced subgraph of
  $G$, for any two vertices $u,v \in V(S)$ with $B(u,S) \neq B(v,S)$ we have $B(u,G) \neq B(v,G)$.  This means that a straight enumeration of
  $G$ implies a straight enumeration of $S$.  Additionally, the
  straight enumeration of each connected component of $G$ is unique up
  to reversal by Proposition~\ref{proper:theorem:uniqueEnumeration}.
  As a result, for any proper interval representation $R$ of $G$, the
  blocks $B_i$ and $B_{i+1}$ are consecutive in $\oo_S(R)$.  This
  holds for any two consecutive blocks in \oo, which means that the
  consecutivity of all blocks of $S$ is fixed for all simultaneous
  representations of $\mathcal{G}$.  As a consequence
  $\oo_S(\mathcal{R})$ is fixed up to complete reversal.
\end{proof}
Let $G$ be a proper interval graph consisting of the connected
components $C_1, \dots, C_k$ with straight enumerations
$\oo_1, \dots, \oo_k$.  Let \cOrd{\conc}{\oo_1}{\cdots}{\oo_k}\eOrd be
a straight enumeration of $G$.  Then we say the straight enumeration
$\oo' = $
\cOrd{\conc}{\oo_1}{\cdots}{\oo_{i-1}}{\oo_i^r}{\oo_{i+1}}{\cdots}{\oo_{k}}\eOrd
is \emph{obtained from \oo by reversal of $C_i$}.  For a sunflower
graph $\mathcal{G}$ containing $G$ with shared graph $S = (V_S, E_S)$,
we call a component $C = (V_C, E_C)$ of $G$ \emph{loose}, if all
vertices $V_S \cap V_C$ are in the same block of $S$.  Reversal of
loose components is the only ``degree of freedom'' among simultaneous
enumerations, besides full reversal, and is formally shown in the appendix.

\begin{figure}[tb]
	\centering
	\includegraphics{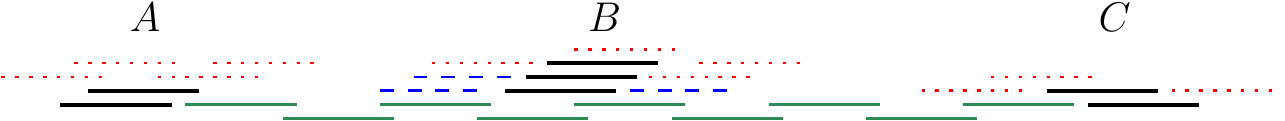}
  \caption{Simultaneous proper interval representation of $G_1$ (green solid), $G_2$ (red dotted), $G_3$ (blue dashed) with shared graph $S$ (black bold). $S$ has three blocks $A$, $B$, $C$.  
  We denote the component of $G_i$ containing a block $D$ by $C_D^i$.
  $C_A^2$, $C_B^2$, $C_B^3$, $C_C^2$ are loose.  $C_A^2$ is independent. $(C_B^2,C_B^3)$ is a reversible part. $(C_C^2)$ is not a reversible part, since $C_C^1$ is aligned at $C$ and not loose.}
	\label{fig:looseDefs}
\end{figure}

To obtain a complete characterization, we now introduce additional
terms to specify which reversals result in simultaneous enumerations (see Figure~\ref{fig:looseDefs}).
Let $\mathcal G = (G_1,\dots,G_k)$ be a connected sunflower proper
interval graph with shared graph $S$.  We say a component $C$ of a
graph in $\mathcal{G}$ \emph{aligns} two vertices $u,v \in S$ if they
are in different blocks of $C$, i.e., $B(u, C)\neq B(v, C)$.  If in addition $u$
and $v$ are in the same block $B$ of $S$, we say $C$ \emph{is oriented
  at} $B$.  If there is another component $C'$ among graphs in
$\mathcal{G}$ oriented at $B$, the orientation of their straight
enumerations in a simultaneous enumeration of $\mathcal{G}$ are
dependent; that is, they cannot be reversed independently.  This is shown formally
in the appendix.
 
For each block $B$ of $S$, let $\mathcal{C}(B)$ be the connected components among graphs in $\mathcal{G}$ oriented at $B$.
Since a component may contain $B$ without aligning vertices, we have $0 \leq |\mathcal{C}(B)| \leq k$.
If $\mathcal{C}(B)$ contains only loose components, we call it a \emph{reversible part}. Note that a reversible part $\mathcal C(B)$ contains at most one component of each graph $G_i$.
Additionally, we call a loose component $C$ \emph{independent}, if it does not align any two vertices of $S$. 
Let $(\sigma_1,\dots,\sigma_k)$ and $(\sigma'_1,\dots,\sigma'_k)$ be tuples of straight enumerations of $G_1,\dots,G_k$. 
We say $(\sigma'_1,\dots,\sigma'_k)$ is \emph{obtained from $(\sigma_1,\dots,\sigma_k)$ through reversal} of reversible part $\mathcal C(B)$, if $\sigma'_1,\dots,\sigma'_k$ are obtained by reversal of all components in $\mathcal C(B)$.

\begin{restatable}[$\star$]{theorem}{simulEnumOperationsTheorem}
  \label{the:simEnumChar}
  Let $\mathcal G = (G_1,\dots,G_k)$ be a connected sunflower graph with shared graph $S$ and simultaneous enumeration $\rho = (\oo_1, \dots, \oo_k)$.
  Then $\rho' = (\oo_1', \dots, \oo_k')$ is a simultaneous enumeration of $\mathcal{G}$ if and only if $\rho'$ can be obtained from $\rho$ or $\rho^r$ through reversal of independent components and reversible parts.
\end{restatable}

\section{Sunflower Unit Interval Graphs}
\label{sec:unit}
In the previous section we characterized all simultaneous enumerations for a sunflower proper interval graph $\mathcal G$.
We say a simultaneous proper/unit interval representation of a sunflower graph $\mathcal G$ \emph{realizes} a simultaneous enumeration $\zeta=(\zeta_1,\dots,\zeta_k)$ of $\zeta$, if for $i\in\{1,\dots,k\}$ the representation of $G_i$ corresponds to the straight enumeration $\zeta_i$.
In Section~\ref{sec:suig_enums} we provide a criterion which determines for a given simultaneous enumeration $\zeta$ of $\mathcal G$ whether there is a simultaneous unit interval representation of $\mathcal G$ that realizes $\zeta$. Namely, the criterion is the avoidance of a certain configuration in a partial vertex order of $\bigcup_\mathcal G$ \emph{induced by} $\zeta$. 
In Section~\ref{sec:recognition} we combine these findings to efficiently recognize simultaneous unit interval graphs.

\subsection{Simultaneous Enumerations of Sunflower Unit Interval Graphs}
\label{sec:suig_enums}
We first obtain a combinatorial characterization by reformulating the problem of finding a representation as a restricted graph sandwich problem~\cite{golumbic1995graph}.

\begin{restatable}[$\star$]{lemma}{Sandwich}
  \label{lem:Hsandwitch} 
  A sunflower graph $\mathcal G$ has a simultaneous unit interval representation that realizes a simultaneous enumeration $\zeta=(\zeta_1,\dots,\zeta_k)$ if and only if there is some graph $H$ with $V(H)=V(\mathcal G)$ that contains the graphs $G_1,\dots,G_k$ as induced subgraphs and has a fine enumeration $\sigma$ such that for $i\in\{1,\dots,k\}$ straight enumeration $\zeta_i$ is compatible with $\sigma$ on $V_i$.
\end{restatable}
Our approach is to obtain more information on what graph $H$ and the fine enumeration $\sigma$ must look like. We adapt a characterization of Looges and Olariu~\cite{looges1993optimal} to obtain four implications that can be used given only partial information on $H$ and $\sigma$ (as given by Lemma~\ref{lem:Hsandwitch}); see Figure~\ref{fig:ExtensionSteps}.   
For the figures in this section we use arrows to represent a partial order between two vertices. We draw them solid green if they are adjacent, red dotted if they are non-adjacent in some graph $G_i$, and black dashed if they may or may not be adjacent.

\begin{theorem}[Looges and Olariu~\cite{looges1993optimal}]
  \label{the:3vertexCondition}
  A vertex order of a graph $H=(V,E)$ is a fine enumeration if and only if for $v,u,w \in V$ with $v<_\sigma u<_\sigma w$ and $vw\in E$ we have $vu, uw\in E$.
\end{theorem}

\begin{restatable}[$\star$]{corollary}{fourVertexCondition}
  \label{cor:4vertexCondition} 
  A vertex order of a graph $H=(V,E)$ is a fine enumeration if and only if there are no four vertices $v,u,x,w \in V$ with $v\le_\sigma u\le_\sigma x\le_\sigma w$ and $vw\in E$ and $ux\not\in E$.
\end{restatable}

\newcommand{\mydigit}{0.16}
  
\begin{figure}[tb]	
\captionsetup[subfigure]{justification=centering}
\begin{subfigure}[t]{0.05\textwidth}
\centering
	\includegraphics{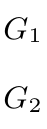}
\end{subfigure}
\begin{subfigure}[t]{\mydigit\textwidth}
\centering
	\includegraphics{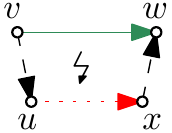}
  \caption{}
  \label{fig:forbiddenConfig}
\end{subfigure}
~
\begin{subfigure}[t]{\mydigit\textwidth}
\centering
	\includegraphics{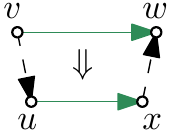}
  \caption{}
  \label{fig:HTExtensionStep_left}
\end{subfigure}
~
\begin{subfigure}[t]{\mydigit\textwidth}
\centering
	\includegraphics{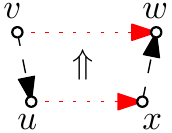}
  \caption{}
  \label{fig:HTExtensionStep_right}
\end{subfigure}
~
\begin{subfigure}[t]{\mydigit\textwidth}
\centering
	\includegraphics{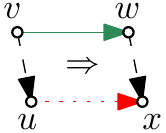}
  \caption{}
  \label{fig:EnumExtensionStep_left}
\end{subfigure}
~
\begin{subfigure}[t]{\mydigit\textwidth}
\centering
	\includegraphics{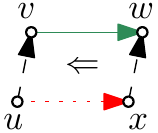}
  \caption{}
  \label{fig:EnumExtensionStep_right}
\end{subfigure}
\caption[]{(a): The forbidden configuration of Corollary~\ref{cor:4vertexCondition}. (b)-(e): The four implications of Corollary~\ref{cor:4implications}.
}\label{fig:ExtensionSteps}
\end{figure}

\begin{corollary}
  \label{cor:4implications}
  Let $H=(V,E)$ be a graph with fine enumeration $\sigma$. Let $v,u,x,w\in V$ and $u \leq_\sigma x$ as well as $v \leq_\sigma w$. Then we have (see Figure~\ref{fig:ExtensionSteps}):
  \begin{compactenum}[(i)]
    \item \label{itm:ext1} $vw\in E \wedge  v\le_\sigma u \wedge x\le_\sigma w \Rightarrow ux\in E$
    \item \label{itm:ext2} $ux\not\in E \wedge v\le_\sigma u \wedge x\le_\sigma w \Rightarrow vw\not\in E$
    \item \label{itm:ext3} $vw\in E \wedge ux\not\in E \wedge v\le_\sigma u \Rightarrow w<_\sigma x$
    \item \label{itm:ext4} $vw\in E \wedge ux\not\in E \wedge x\le_\sigma w \Rightarrow u<_\sigma v$.
  \end{compactenum}
\end{corollary}

Now we introduce the forbidden configurations for \emph{simultaneous} enumerations of sunflower unit interval graphs.
Throughout this section let $\mathcal G = (G_1,\dots,G_k)$ be a sunflower graph with shared graph $S$ and simultaneous enumeration $\zeta=(\zeta_1,\dots,\zeta_k)$. 
Furthermore, let $V_i = V(G_i)$ and $E_i = E(G_i)$, for $i\in\{1,\dots,k\}$. Finally, let $V=V_1\cup\dots\cup V_k$.   
For a straight enumeration $\eta$ of some graph $H$ we say for $u,v\in V(H)$
that $u<_\eta v$, if $u$ is in a block before $v$, and we say $u\le_\eta v$, if
$u=v$ or $u<_\eta v$. We call $\le_\eta$ the \emph{partial order on $V(H)$
corresponding to $\eta$}. Note that for distinct $u,v$ in the same block we have neither $u>_\eta v$ nor $u\le_\eta v$. We write $u\le_i v$ and $u<_i v$ instead of $u\le_{\zeta_i} v$ and $u<_{\zeta_i} v$, respectively.

\begin{figure}[tb]
	\centering
	\includegraphics{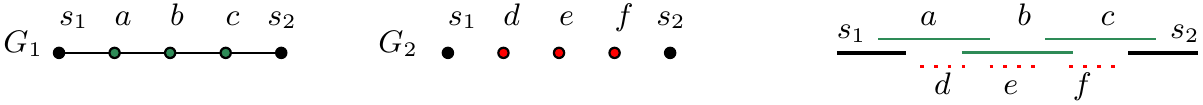}
	\caption{A sunflower graph $\mathcal{G}=(G_1,G_2)$ with shared vertices $s_1,s_2$ (black, bold). Let $\zeta$ be the simultaneous enumeration realized by the given simultaneous proper interval representation. In $(G_1,\zeta_1)$ we have the $(s_1,s_2)$-chain $C=(s_1,a,b,c,s_2)$ of size $5$ (green, solid). In $(G_2,\zeta_2)$ we have the $(s_1,s_2)$-bar $B=(s_1,d,e,f,s_2)$ of size $5$ (red, dotted). Hence, sunflower graph $\mathcal G$ has the conflict $(C,B)$ for the simultaneous enumeration $\zeta$.}
	\label{fig:conflict}
\end{figure}

Let $u,v\in V(S)$ with $u\ne v$.
 For $i\in\{1,\dots,k\}$ a \emph{$(u,v)$-chain} of \emph{size} $m\in\mathds N$ in $(G_i,\zeta_i)$ is a sequence $(u=c_1,\dots,c_m=v)$ of vertices in $V_i$ with $c_1<_i\dots<_i c_m$ that corresponds to a path in $G_i$.
A \emph{$(u,v)$-bar} between $u$ and $v$ of \emph{size} $m\in\mathds N$ in $(G_i,\zeta_i)$ is a sequence $(u=b_1,\dots,b_m=v)$ of vertices in $V_i$ with $b_1<_i\dots<_i b_m$ that corresponds to an independent set in $G_i$.
An example is shown in Figure~\ref{fig:conflict}.
If there is a $(u,v)$-chain $C$ in $G_i$ of size $l\ge 2$ and a $(u,v)$-bar $B$ in $(G_j,\zeta_j)$ of size at least $l$, then we say that $(C,B)$ is a \emph{$(u,v)$-(chain-bar-)conflict} and that $\mathcal G$ \emph{has conflict $(C,B)$ for $\zeta$}.
Note that one can reduce the size of a larger $(u,v)$-bar by removing intervals between $u$ and $v$. Thus, we can always assume that in a conflict, we have a bar and a chain of the same size $l\ge 2$.
Assume $\mathcal G$ has a simultaneous unit interval representation realizing $\zeta$. 
If a graph $G\in\mathcal G$ has a $(u,v)$-chain of size $l\ge 2$, then the
distance between the intervals $I_u,I_v$ for $u,v$ is smaller than $l-2$. On the
other hand, if a graph $G\in\mathcal G$ has a $(u,v)$-bar of size $l$, then the distance between $I_u,I_v$ is greater than $l-2$. Hence, sunflower graph $\mathcal G$ has no conflict.
The result of this section is that the absence of conflicts is not only necessary, but also sufficient.

\begin{restatable}{theorem}{SUIRchar}
  \label{the:SUIRchar}
  Let $\mathcal G$ be a sunflower proper interval graph with simultaneous enumeration $\zeta$.
  Then $\mathcal G$ has a simultaneous unit interval representation that realizes $\zeta$ if and only if $\mathcal G$ has no conflict for $\zeta$.
\end{restatable}

Recall that $V_i=V(G_i)$ for $i\in\{1,\dots,k\}$ and $V=V_1\cup\dots\cup V_k$.
Let $\alpha^\star$ be the union of the partial orders on $V_1,\dots,V_k$ corresponding to $\zeta_1,\dots,\zeta_k$.
Set $\alpha$ to be the transitive closure of $\alpha^\star$. We call $\alpha$ the \emph{partial order on $V$ induced by $\zeta$}.
The rough idea is that the partial order on $V$ induced by the simultaneous enumeration $\zeta$ is extended in two sweeps to a fine enumeration of some graph $H$ that contains $G_1,\dots,G_k$ as induced subgraphs; see Figures~\ref{fig:scout_fail}, \ref{fig:scout_success}. 
For $(u,v)\in \alpha$ we consider $u$ to be to the left of $v$.
The first sweep (\emph{scouting}) goes from the right to the left and makes only necessary extensions according to Corollary~\ref{cor:4implications}\,($\ref{itm:ext4}$). If there is a conflict, then it is found in this step. Otherwise, we can greedily order the vertices on the way back by additionally respecting Corollary~\ref{cor:4implications}\,($\ref{itm:ext3}$) (\emph{zipping}) to obtain a linear extension where both implications are satisfied. In the last step we decide which edges $H$ has by respecting Corollary~\ref{cor:4implications}\,($\ref{itm:ext1}$).

\begin{figure}[tb]
	\centering
	\includegraphics[scale=1]{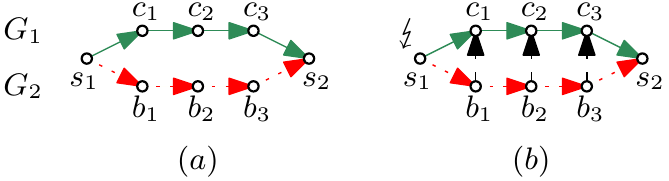}
  \caption{(a): A simultaneous enumeration with conflict. (b): Result with added orderings after scouting, starting at $s_2$ and finding the conflict in $s_1$.}
	\label{fig:scout_fail}
\end{figure}

\begin{figure}[tb]
	\centering
	\includegraphics[scale=1]{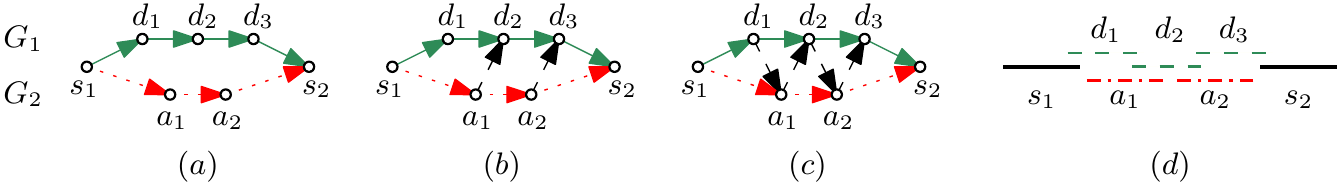}
	\caption{(a): A simultaneous enumeration without conflict. (b): Result with added orderings after scouting. (c): Resulting linear order after zipping. Note that $a_1$ comes before $d_2$ in the linear order thanks to scouting. Choosing otherwise would imply a contradiction at $s_2$. (d): Resulting unit interval representation for the sandwich graph.}
	\label{fig:scout_success}
\end{figure}

For $h\in \{1,\dots,k\}$ we say two vertices $u,v\in V_h$ are \emph{indistinguishable in $\mathcal G$} if we have $N_{G_i}(u) = N_{G_i}(v)$ for all $i\in \{1,\dots,k\}$ with $u,v\in V_i$. In that case $u$, $v$ can be represented by the same interval in any simultaneous proper interval representation. Thus, we identify indistinguishable vertices. If $u,v\in V_h$ are not indistinguishable, then we have $N_{G_j}(u)\ne N_{G_j}(v)$ for some $j\in\{1,\dots,k\}$. In that case $u$, $v$ are ordered by $\zeta_j$ and therefore by $\alpha$. That is, we can assume $\alpha$ to be a linear order on $V_i$ for $i\in\{1,\dots,k\}$.
Note that $u$, $v$ may be ordered even if they are indistinguishable in some input graphs.

 For $i\in\{1,\dots,k\}$, let $G_i^c=(V_i,\binom{V_i}{2}\setminus E_i)$ be the
 complement of $G_i$.  We set $E=\left\{(u,v)\in \alpha\mid uv\in E_1\cup\dots\cup E_k\right\}$ and $F=\left\{(u,v)\in \alpha\mid uv\in E(G_1^c\cup\dots\cup G_k^c)\right\}$.
We call a partial order $\sigma$ on $V$ \emph{left-closed} if we have
\begin{equation}\forall v,w,u,x\in V\colon ~(vw\in E~\wedge~ ux\in F~\wedge~
  x\le_\sigma w)~\Rightarrow~ u<_\sigma v. \label{eq:left-closed} \end{equation} 
  Note that a fine enumeration of a graph $H$ with $G_1,\dots,G_k$ as induced subgraphs is left-closed by Corollary~\ref{cor:4implications}\,($\ref{itm:ext4}$).
We describe the result of the first sweep with the following lemma.

\begin{lemma}
  \label{lem:leftClosedChar}
  A sunflower graph $\mathcal G$ has no conflict for a simultaneous enumeration $\zeta$ if and only if there is a left-closed partial order $\tau$ that extends the partial order on $V(\mathcal G)$ induced by $\zeta$.
\end{lemma}
{\renewcommand{\proofname}{Proof Sketch} \begin{proof}
  If there is a conflict $(C,B)$, then the partial order $\alpha$ induced by $\zeta$ cannot be extended to be left-closed since then for $i\in\{1,\dots, k-1\}$ the $i$'th vertex of $C$ and $B$ must be ordered and distinct while the first vertex is shared; see Figure~\ref{fig:scout_fail}.

  Otherwise, we process the vertices from the right to the left and add for each of them the implied orderings (each is considered as vertex $x$ in the definition of left-closed). 
  First consider the case of just two input graphs $G_1, G_2$.
  Let $X$ be the set of already processed vertices and let $\sigma$ be the current partial order. We next process a maximal vertex $x\in V\setminus X$. Let $x\in V _i$. Then we choose $u_i$ to be the rightmost vertex in $V_i$ with $u_ix\in F$ and for $j\ne i$ we choose $w_j$ to be the leftmost vertex in $V_j$ with $x\le w_j$ and $v_j$ to be the leftmost vertex in $V_j$ with $v_jw_j\in E$; see Figure~\ref{fig:scoutExtension}. Each of $u_i,v_j,w_j$ may not exist. If they do, we extend $\sigma$ to $\sigma'$ by adding the ordering $u_i\le_{\sigma'} v_j$. The other implied orderings are exactly those obtained by transitive closure.

  \begin{figure}[tb]	
    \captionsetup[subfigure]{justification=centering}
    \centering
    \begin{subfigure}[b]{0.4\textwidth}
      \includegraphics{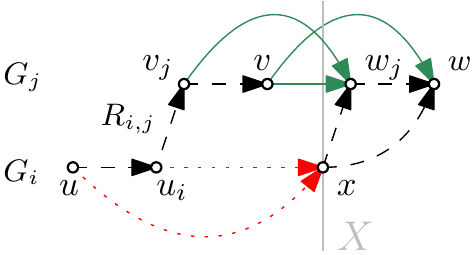}
      \caption{}
      \label{fig:scoutExtension}
    \end{subfigure}
    ~
    \begin{subfigure}[b]{0.17\textwidth} 
      \includegraphics{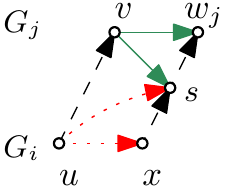}
      \caption{}
      \label{fig:CPpairStart}
    \end{subfigure}
    ~
    \begin{subfigure}[b]{0.2\textwidth} 
      \includegraphics{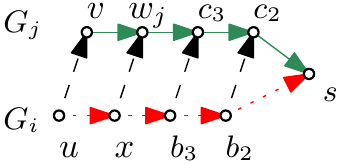}
      \caption{}
      \label{fig:CBpairStep}
    \end{subfigure}
    \caption{(a): The vertices $u_i,v_j,w_j$ as derived from $x$ and $X$. The introduced ordering $(u_i,v_j)$ is marked with $R_{i,j}$. 
    (b),(c): Both cases of a chain-bar pair for $u$ and $v$.
    \label{fig:scouting1}}  
  \end{figure}

  Two vertices $u\in V_1, v\in V_2$ are only ordered by $\alpha$ if there is a shared vertex $s$ with $u\le_\alpha s \le_\alpha v$ or $v\le_\alpha s \le_\alpha u$. The key observation is that if $u$ is ordered before $v$ due to a necessary extension, then there is a shared vertex $s$ and a $(v,s)$-chain and a $(u,s)$-bar of equal size (\emph{chain-bar pair}): If we have $x\le_\alpha w_j$, then there is a shared vertex $x\le_\alpha s\le_\alpha w_j$ and by Theorem~\ref{the:3vertexCondition} we obtain $us\in F$ and $vs\in E$, which yields a chain-bar pair; see Figure~\ref{fig:CPpairStart}. Otherwise we have a chain-bar pair for $x$ and $w_j$ that can be extended by $u$ and $v$; see Figure~\ref{fig:CBpairStep}. With the absence of conflicts this ensures that vertices ordered according to the left-closed property are actually distinct. 

  \begin{figure}[tb]	
\captionsetup[subfigure]{justification=centering}
\centering
\begin{subfigure}[b]{0.45\textwidth}
	\includegraphics{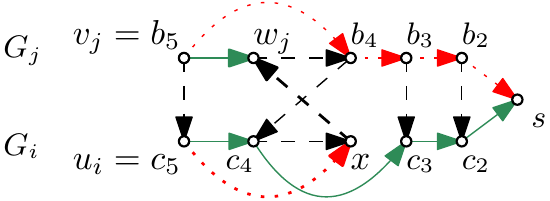}
	\caption{}
  \label{fig:scout_antisymmetry}
\end{subfigure}
~
\begin{subfigure}[b]{0.3\textwidth} 
	\includegraphics{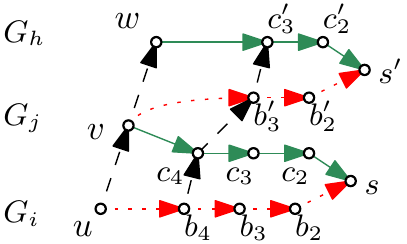}
	\caption{}
  \label{fig:scout_trans}
  \end{subfigure} 
  \caption{(a): Example situation for $(v_j,u_i)\in\tau_{i,j}\setminus\alpha^\star$. We have the $(v_j,s)$-bar $(v_j,b_4,b_3,b_2,s)$ and the $(u_i,s)$-chain $(u_i,c_4,c_3,c_2,s)$ and obtain $x\le_{\tau_{i,j}} w_j\le_{\tau_{i,j}} b_4\le_{\tau_{i,j}} c_4\le_{\tau_{i,j}} x$. (b): Example situation for the transitivity of $\tau$ where we have a chain-bar pair for $u,v$ as well as for $v,w$. We obtain $b_4 \le_{\tau_{i,j}} c_4\le_{\tau_{i,j}} b_3'\le_{\tau_{i,j}} c'_3$ and since $u<_\alpha b_4$ and $w<_\alpha c_3'$ we get $b_4\le_{\tau_{i,h}} c_3'$ in an appropriate induction and with $\tau_{i,h}$ being left-closed we obtain $u\le_{\tau_{i,h}} w$. (The base cases for the induction involve shared vertices and thereby only two input graphs.)}
\end{figure}

  Assume a new extension would violate the property of antisymmetry. This would mean we already had $v_j<_\sigma u_i$, which would imply a cyclic ordering of $x$, $w_j$ with elements of the (necessary) chain-bar pair for $v_j$, $u_i$ in a prior step; see Figure~\ref{fig:scout_antisymmetry}. 
  Finally, for more than two input graphs we obtain a corresponding ordering $\tau_{i,j}$ for each pair of input graphs $G_i,G_j$. Let $\tau=\bigcup_{i,j\in\{1,\dots,k\}} \tau_{i,j}$ be their union.  For $u<_{\tau_{i,j}}v <_{\tau_{j,h}} w$ we can prove $u <_{\tau_{i,h}} w$ by using chain-bar pairs and induction; see Figure~\ref{fig:scout_trans}. Hence, $\tau$ is already transitive and the other properties are easy to verify.  
\end{proof}}

By respecting the orderings obtained by scouting we avoid wrong decisions when greedily adding vertices to a linear ordering in the zipping step; see Figure~\ref{fig:scout_success}. 

\begin{lemma}
  \label{lem:zipping}
 Let $\mathcal G$ be a sunflower graph with a simultaneous enumeration $\zeta$. There is a left-closed linear order $\tau$ that extends the partial order $\alpha$ on $V(\mathcal G)$ induced by $\zeta$ if and only if there is a left-closed partial order $\sigma\supseteq\alpha$.
\end{lemma}
{\renewcommand{\proofname}{Proof Sketch} \begin{proof}
  Given $\sigma$ we process the vertices from the left to the right. We add in each step a leftmost vertex $u$ of the remaining vertices to a set $U$ of the processed vertices that are linearly ordered. We denote the current order by $\sigma'$. Vertex $u$ is then ordered before all other vertices in $V\setminus U$. To avoid that the left-closed property is violated when adding such orderings for another vertex, we ensure our extended order $\sigma''\supseteq \sigma'$ is \emph{right-closed on $U$} meaning that
  \begin{equation}
    \forall u,v\in U, w,x\in V\colon (vw\in E \wedge ux\in F\wedge v\le u) \Rightarrow w< x.
  \end{equation}

  \begin{figure}[tb]	
\captionsetup[subfigure]{justification=centering}
\centering
\begin{subfigure}[b]{0.3\textwidth}
	\includegraphics{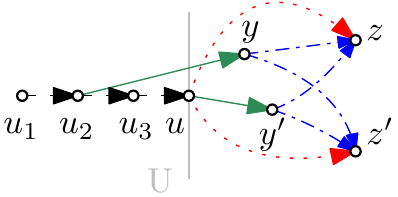}
	\caption{}
  \label{fig:Rzip}
\end{subfigure}
~
\begin{subfigure}[b]{0.3\textwidth} 
	\includegraphics{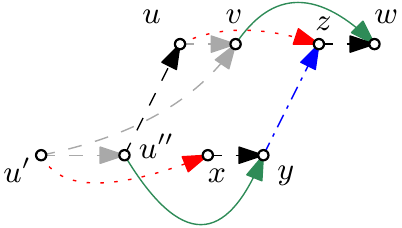}
	\caption{}
  \label{fig:RleftClosed}
  \end{subfigure} 
  \caption{(a): orderings added during a zipping step (blue dash-dotted). All vertices in $Y=\{y,y'\}$ are ordered before those in $Z=\{z,z'\}$. (b): The case for $\sigma''$ being left-closed where we have $x\le_{\sigma''} w$ due to transitivity. This means there is some ordering $(y,z)\in Y\times Z$ with $x\le_{\sigma'} y \le_{\sigma''} z \le_{\sigma'} w$. We further have a vertex $u''\in U$ with $u''y\in E$ and $uz\in F$. 
  Given vertices $u',v\in V$ with $u'x\in F$ and $vw\in E$ we obtain $u'<_{\sigma'} u''$ and $u<_{\sigma'} v$ since $\sigma'$ is left-closed. This yields $u'<_{\sigma''} v$.}
\end{figure}

  To this end, we consider the current vertex $u$ as vertex $u$ in the definition of right-closed and add all implied orderings in $\sigma''$. This means for each vertex $y\in Y = \{y\in V\mid \exists u'\in U\colon uy\in E\}$ and each vertex $z\in Z=\{z\in V\mid uz\in F\}$ we set $y\le_{\sigma''} z$; see Figure~\ref{fig:Rzip}. We further extend $\sigma''$ to be transitive. Note that there are no two vertices $y\in Y, z\in Z$ with $y\le_\sigma z$, since $\sigma$ is left-closed and for $u'\in U$ we have $u'\le_\sigma u$. With this observation we can verify that $\sigma''$ is antisymmetric and left-closed; see Figure~\ref{fig:RleftClosed}.
\end{proof}}

Finally, we construct a graph $H=(V,E')$ for which the obtained linear order $\tau$ is a fine enumeration. We do so by setting $E'=\{ux\in V^2 \mid \exists vw\in E\colon v\le_\tau u <_\tau x\le_\tau w\}$ in accordance with Corollary~\ref{cor:4implications}\,~($\ref{itm:ext1}$).

\begin{lemma}[$\star$]
  \label{lem:sandwich}
 Let $\mathcal G=(G_1,\dots,G_k)$ be a sunflower graph with a simultaneous enumeration $\zeta$. A linear order $\tau$ that extends the partial order on $V(\mathcal G)$ induced by $\zeta$ is a fine enumeration for some graph $H$ that has $G_1,\dots,G_k$ as induced subgraphs if and only if $\tau$ is left-closed.
\end{lemma}

Combining Lemmas~\ref{lem:Hsandwitch}, \ref{lem:leftClosedChar}, \ref{lem:zipping} and \ref{lem:sandwich} we obtain Theorem~\ref{the:SUIRchar}. 
	
\SUIRchar*

\subsection{Recognizing Simultaneous Unit Interval Graphs in Polynomial Time}
\label{sec:recognition}
With Theorems~\ref{the:simEnumChar} and~\ref{the:SUIRchar} we can now efficiently recognize simultaneous unit interval graphs.

\begin{restatable}{theorem}{SUIRalg}
  Given a sunflower graph $\mathcal G=(G_1,\dots,G_k)$, we can decide in $O(|V|\cdot|E|)$ time, whether $\mathcal G$ is a simultaneous unit interval graph, where $(V,E)=G^\star=G_1\cup\dots\cup G_k$. If it is, then we also provide a simultaneous unit interval representation in the same time. 
\end{restatable}
{\renewcommand{\proofname}{Proof Sketch} \begin{proof}
  Here we establish polynomial time recognition, and the stated time is proven in the appendix. As discussed earlier, we can assume that $G^\star$ is connected. With Theorem~\ref{the:spirRecognition} we obtain a simultaneous enumeration $\zeta$ of $\mathcal G$, unless $\mathcal G$ is not a simultaneous proper interval graph. 
By Theorem~\ref{the:SUIRchar}, the sunflower graph $\mathcal G$ is a simultaneous unit interval graph if and only if there is a simultaneous enumeration $\eta$ for which $\mathcal G$ has no conflict. In that case $\eta^r$ also has no conflict.
With Theorem~\ref{the:simEnumChar} we have that $\eta$ or $\eta^r$ is obtained from $\zeta$ by reversals of reversible parts and independent components. Hence, we only need to consider simultaneous enumerations obtained that way. 

Since every single graph $G_i$ is proper, it has no conflict and we only need to consider $(u,v)$-conflicts with $u,v\in V(S)$, where $S$ is the shared graph. 
The minimal $(u,v)$-chains for $G_i$ are exactly the shortest paths in $G_i$ and thus independent from reversals.
On the other hand, for the maximal size of $(u,v)$-bars in $G_i$ only the reversals of the two corresponding components $C,D$ of $u,v$ are relevant, while components in-between always contribute their maximum independent set regardless of whether they are reversed. We can thus compute for $i,j\in\{1,\dots,k\}$, $u,v\in V(S)$ and each of the four combinations of reversal decisions (reverse or do not reverse) for the corresponding components $C,D$ of $u,v$, whether they yield a conflict at $(u,v)$.
  We can formulate a corresponding $2$-SAT formula $\mathcal F$:
  For every independent component and every reversible part we introduce a literal that represents whether it is reversed or not. For every combination of two reversal decisions that yields a conflict we add a clause that excludes this combination. If $\mathcal F$ is not satisfiable, then every simultaneous enumeration yields a conflict. Otherwise, a solution yields a simultaneous enumeration without conflict. We obtain a simultaneous unit interval representation by following the construction in Section~\ref{sec:suig_enums}. 
\end{proof}}

\section{Conclusion}
\label{ch:conclusion}

We studied the problem of simultaneous representations of proper
and unit interval graphs.  We have shown that, in the sunflower case,
both simultaneous proper interval graphs and simultaneous unit
intervals can be recognized efficiently.  While the former can be
recognized by a simple and straightforward recognition algorithm, the
latter is based on the three ingredients: 1) a complete
characterization of all simultaneous proper interval representations
of a sunflower simultaneous graph, 2) a characterization of the
simultaneous proper interval representations that can be realized by a
simultaneous unit interval representation and 3) an algorithm for
testing whether among the simultaneous proper interval representations
there is one that satisfies this property.

\subparagraph{Future Work.}
While our algorithm for (sunflower) simultaneous proper interval graphs has optimal linear running time, we leave it as an open problem
whether simultaneous unit interval graphs can also be recognized in linear time.  

Our main open question is about the complexity of sunflower
simultaneous interval graphs.  Jampani and Lubiw~\cite{Jampani2010}
conjecture that they can be recognized in polynomial time for any
number of input graphs.  However, even for three graphs the problem is
still open.

\bibliographystyle{plainurl}
\bibliography{references}

\appendix
\addtocontents{toc}{\protect\setcounter{tocdepth}{1}}
\section{Omitted Proofs from Section~\ref{ch:preliminaries}}

\unionconnected*

\begin{proof}
  Clearly, a simultaneous (proper) interval representation $\mathcal
  R$ of $\mathcal G$ induces a representation for each $\mathcal G_i$.
  Conversely, given simultaneous (proper) interval representations
  $\mathcal R_i$ of $\mathcal G_i$ for $i=1,\dots,l$, we can combine
  them such that all intervals in $\mathcal R_i$ are placed to the
  right of all intervals in $\mathcal R_{i-1}$ for $i=2,\dots,l$ to
  obtain a simultaneous (proper) interval representation $\mathcal R$
  of $\mathcal G$.
\end{proof}

\section{Formal Proofs for Sunflower Proper Interval Graphs}
\label{app:proofs:simulProp}
\subsection{Proof of Lemma~\ref{lem:EnumCombination}}
\label{app:proof:enumCombination}
\enumCombinationLemma*
\begin{proof}
  Assume $\mathcal{R} = (R_1, \dots, R_k)$ is a simultaneous proper interval representation of $\mathcal{G}$ with corresponding straight enumerations $\oo_i = \oo(R_i)$.
  Let \co be a linear order of the vertices $V_S$ according to their left endpoints in $\mathcal{R}$, breaking ties arbitrarily.
  We claim that each $\oo_i$ is compatible with $\co$.
  If $\oo(R_i)$ is not compatible with $\co$, there exist vertices \cOrd{<_\co}{u}{v}\eOrd such that \cOrd{<_{\oo_i}}{B(v, G_i)}{B(u, G_i)}\eOrd.
  By the definition of extracted straight enumeration, this implies that the interval of $v$ has its left endpoint before the interval of $u$ in $R_i$, which contradicts \cOrd{<_\co}{u}{v}\eOrd.

  \begin{figure}[tb]
    \centering
    \includegraphics[width=\linewidth]{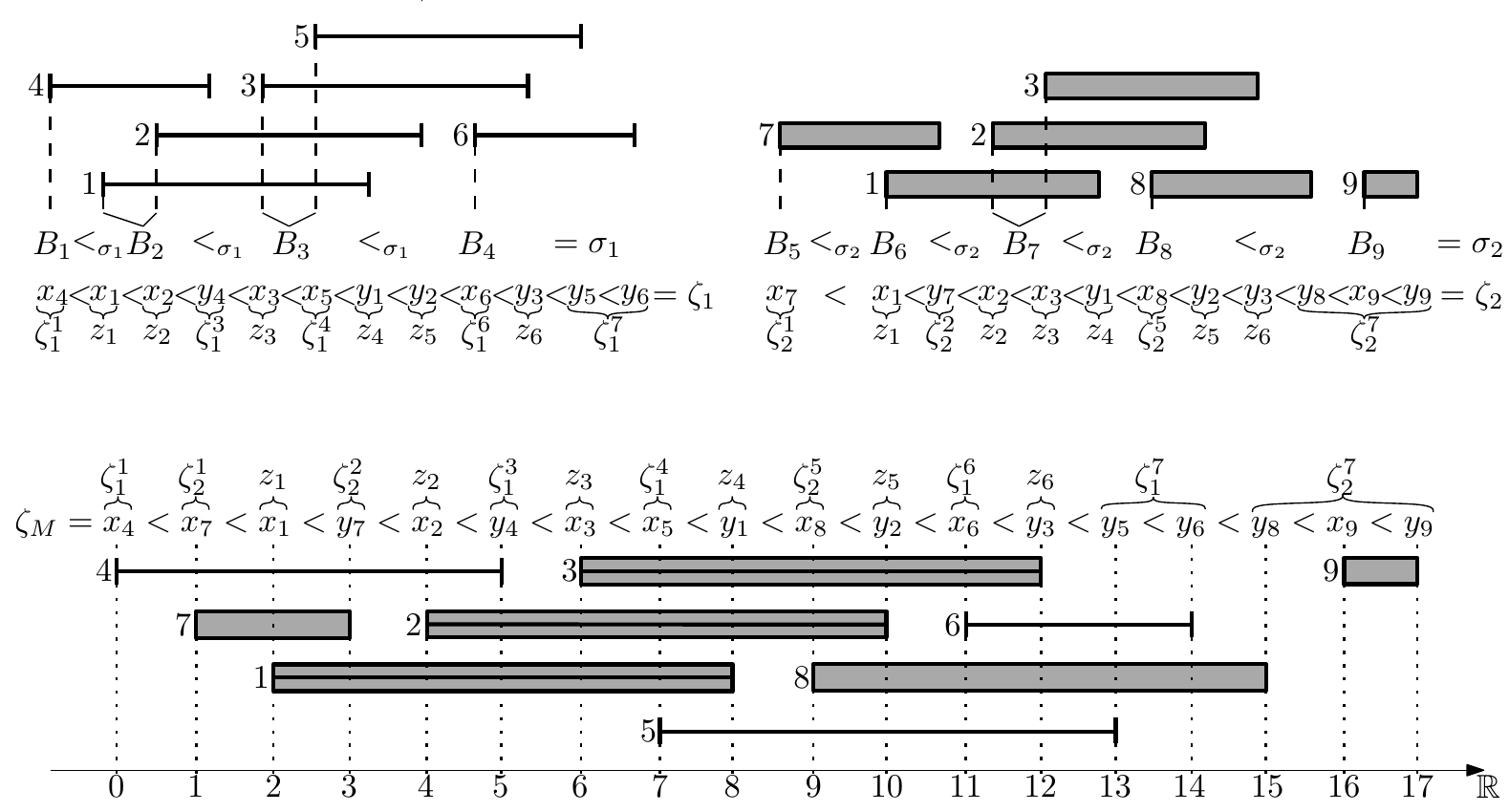}
    \caption{An illustration of the construction in the proof of Lemma~\ref{lem:EnumCombination}.
    Top: Graph $G_1$ on the left, Graph $G_2$ on the right. They share the clique $\{1, 2, 3\}$ and their straight enumerations $\oo_1, \oo_2$ are compatible with the ordering of shared vertices $1 < 2 < 3$.
    The interval representations of $G_1$ and $G_2$ are shown as visualization and are not a determining factor in the construction of $\co_1$ and $\co_2$.
    Bottom: The constructed simultaneous interval representation $\mathcal{R} = (R_1, R_2)$ with the linear order $\oo_M$ from which $\mathcal{R}$ is derived. The representation $R_1$ is illustrated by lines while $R_2$ is illustrated by boxes.}
    \label{simul:fig:megaOrder}
  \end{figure}

  Conversely, we show how to construct a simultaneous proper interval
  representation $\mathcal{R}$ of $\mathcal{G}$ using a linear order \co of
  $V_S$ and straight enumerations $\oo_i$ of each $G_i \in \mathcal{G}$ compatible with \co.
  An illustration of the following construction is given in Figure \ref{simul:fig:megaOrder}.
  For each graph $G_i = (V_i, E_i) \in \mathcal{G}$ with $V_i = \{v_1, \dots,
  v_q\}$, we construct a set of interval endpoints $M_i = \{x_1, y_1, \dots ,
  x_q, y_q \}$ and a linear order $\zeta_i$ of $M_i$ as follows.
  For any two vertices $v_j, v_l \in V_i$ we define $\co_i$ as
  \begin{align}
    \label{simul:function:intervalOrder}
    \text{\cOrd{\leq_{\co_i}}{x_j}{x_l}\eOrd and \cOrd{\leq_{\co_i}}{y_j}{y_l}\eOrd} \Leftrightarrow \; &
    \begin{cases}
      \text{\cOrd{<_{\oo_i}}{B(v_j, G_i)}{B(v_l, G_i)}\eOrd or}\\
      B(v_j, G_i) = B(v_l, G_i) \text{ and } j \leq l
    \end{cases}
    \\ \nonumber \\
    \label{simul:function:intervalIntersection}
    \begin{split}
      \text{\cOrd{\leq_{\co_i}}{x_j}{y_l}\eOrd} \Leftrightarrow \; &\text{\cOrd{\leq_{\oo_i}}{B(v_j, G_i)}{B(v_l, G_i)}\eOrd or } v_j \in N[v_l]
      \\
      \text{\cOrd{\leq_{\co_i}}{y_l}{x_j}\eOrd} \Leftrightarrow \; &\text{\cOrd{<_{\oo_i}}{B(v_l, G_i)}{B(v_j, G_i)}\eOrd and } v_j \notin N[v_l].
    \end{split}
  \end{align}
  Since $\oo_i$ is compatible with \co, it is clear that the linear order $\co'$ = \cOrd{<_{\co'}}{z_1}{\cdots}{z_{2p}}\eOrd of $\{x_j \in M_i, y_j \in M_i : w_j \in V_S\}$ extending to $\co_i$ is the same for all $i \in \{1, \dots, k\}$.
  Therefore, each linear order $\co_i$ can be represented as $\co_i $ = \cOrd{\conc}{\co_i^1}{z_1}{\co_i^2}{\cdots}{z_{2p}}{\co_i^{2p+1}}\eOrd, which allows us to combine all $\co_i$ as follows.
  For the union set $M = \bigcup_{i = 1}^k M_i$, which contains the interval endpoints of vertices in $V_S$ only once, we construct the linear order $\co_M$ = \cOrd{\conc}{\co_1^1}{\cdots}{\co_k^1}{z_1}{\co_1^2}{\cdots}{\co_k^{2p+1}}\eOrd.
  Additionally, for each $m \in M$ we define $\#_M(m)$ as the number of elements before $m$ in $\co_M$.
  That is, $\#_M(m) = |\{m' \in M : \text{\cOrd{<_{\co_M}}{m'}{m}\eOrd}\}|$.
  This allows us to construct the distinct interval representations $R_i = \{[\#_M(x_v), \#_M(y_v)] : v \in V_i \}$, which yields the simultaneous interval representation $\mathcal{R} = (R_1,\dots,R_k)$.
  It remains to show that each $R_i$ is a proper interval representation and that it is an interval representation of the respective graph $G_i$.

  Suppose for sake of contradiction that there exists an interval representation $R_i \in \mathcal{R}$ that is not a proper interval representation.
  Then there exist two vertices $v_j \neq v_l \in V_i$ such that the interval of $v_j$ properly contains the interval of $v_l$ in $R_i$.
  By construction of $R_i$ this implies \cOrd{<_{\co_i}}{x_j}{x_l}\eOrd and \cOrd{<_{\co_i}}{y_l}{y_j}\eOrd.
  Using the construction rules (\ref{simul:function:intervalOrder}), this is only possible if $B(v_j, G_i) = B(v_l, G_i)$ and $j = l$, which contradicts the assumption $v_j \neq v_l$.

  To show that each interval representation $R_i \in \mathcal{R}$ is a representation of $G_i \in \mathcal{G}$, we show that $R_i$ models exactly the edges of $G_i$.
  For any two vertices $v_j,v_l \in V_i$ with $\{v_j,v_l\} \notin E_i$ we have $v_j \notin N[v_l]$ and thus $B(v_j, G_i) \neq B(v_l, G_i)$.
  Without loss of generality, assume \cOrd{<_{\oo_i}}{B(v_j, G_i)}{B(v_l, G_i)}\eOrd.
  By (\ref{simul:function:intervalIntersection}) it follows that \cOrd{<_{\co_M}}{y_j}{x_l}\eOrd, which means that the interval of $v_j$ ends before the interval of $v_l$ begins in $R_i$.
  Conversely, let $v_j,v_l \in V_i$ share an edge $\{u_j,v_l\} \in E_i$.
  By (\ref{simul:function:intervalIntersection}) it follows that \cOrd{\leq_{\co_M}}{x_j}{y_l}\eOrd and \cOrd{\leq_{\co_M}}{x_l}{y_j}\eOrd.
  This means that both intervals begin before either of them ends in $R_i$.
  Therefore each $R_i \in \mathcal{R}$ is a proper interval representation of $G_i \in \mathcal{G}$.  Note that by construction $\oo(R_i) = \oo_i$.
\end{proof}

\subsection{Proof of Theorem~\ref{the:simulEnums}}
\simulEnumerationstheorem*
\begin{proof}
  Let $S = (V_S, E_S)$ be the shared graph of $\mathcal{G}$.
  If $\mathcal{G}$ is a simultaneous proper interval graph, there exist straight enumerations $\oo_i$ of $G_i$ that are compatible with a linear order \co of $V_S$ by Lemma~\ref{lem:EnumCombination}.
  By definition of compatible, for any vertices \cOrd{\leq_\co}{u}{v}\eOrd we
  have \cOrd{\leq_{\oo_i}}{B(u,G_i)}{B(v,G_i)}\eOrd for each $i \in \{1,\dots, k\}$.
  Since \co and each straight enumerations are linear orders, we have \cOrd{\leq_{\oo_i}}{B(u,G_i)}{B(v,G_i)}\eOrd $\Leftrightarrow$ \cOrd{\leq_{\oo_j}}{B(u,G_j)}{B(v,G_j)}\eOrd for any $i,j \in \{1, \dots, k\}$.
  As a result, $(\oo_1, \dots, \oo_k)$ is a simultaneous enumeration of $\mathcal{G}$.

  Conversely, let $(\oo_1, \dots, \oo_k)$ be a simultaneous enumeration of $\mathcal{G}$.
  Then $\co_P = \{(u,v) \in V_S \times V_S \mid \exists i: B(u,G_i) <_{\oo_i} B(v, G_i)\}$ is a partial ordering of $V_S$.
  Let \co be a linear order of $V_S$ extending $\co_P$.
  By definition of simultaneous enumerations and construction of $\co_P$, each straight enumeration $\oo_i$ is compatible with \co.
  Then $\mathcal{G}$ is a simultaneous proper interval graph by Lemma~\ref{lem:EnumCombination}.
  Additionally, the construction in Lemma~\ref{lem:EnumCombination} yields a simultaneous proper interval representation $\mathcal{R} = (R_1, \dots, R_k)$ with $\oo(R_i) = \oo_i$ for each $R_i \in \mathcal{R}$.
\end{proof}
\subsection{Reversal of Loose Components}
\label{app:lemma:onlyLooseReversals}
\begin{lemma}
  \label{lem:SPIR_seOnly1blockRev}
  Let $\mathcal G = (G_1,\dots,G_k)$ be a sunflower proper interval graph with shared graph $S$ and simultaneous enumeration $(\oo_1, \dots, \oo_k)$.
  Then for any simultaneous enumeration $(\oo_1', \dots, \oo_k')$ of $\mathcal{G})$, any straight enumeration $\oo_i'$ can be obtained from $\oo_i$ or $\oo_i^r$ through reversal of loose components.
\end{lemma}
\begin{proof}
  By Theorem~\ref{the:simulEnums}, let $\mathcal{R} = (R_1, \dots, R_k)$ and $\mathcal{R'} = (R'_1, \dots, R'_k)$ be simultaneous proper interval representations of $\mathcal{G}$ with $R_i = \oo_i$ and $R'_i = \oo'_i$.
  Then by Lemma~\ref{simul:lemma:fixSubgraphOrdering} $\oo_S(\mathcal{R'})$ is either equal to $\oo_S(\mathcal{R})$ or $\oo_S(\mathcal{R})^r$.
  We now show that $\oo_i'$ can be obtained from $\oo_i$ through reversal of loose components if $\oo_S(\mathcal{R'}) = \oo_S(\mathcal{R})$.

  For each connected component of $G_i$, its straight enumeration is unique up to reversal by Theorem~\ref{proper:theorem:uniqueEnumeration}.
  This means that $\oo_i$ and $\oo'_i$ can only differ through reversal and reordering of straight enumerations of individual components of $G$.
  Since we can assume that the union graph $\bigcup_{\mathcal{G}}$ is connected by Lemma~\ref{lem:union-connected}, each component of $G$ contains at least one vertex and thus block of $S$.
  Since we have $\oo_S(R_i) = \oo_S(R'_i)$, the order of blocks of $S$ in $\oo_i$ and $\oo'_i$ is identical.
  This means that the straight enumerations of components of $G$ are ordered identically in $\oo_i$ and $\oo'_i$ and that no straight enumeration containing vertices from more than one block of $S$ is reversed.
  As a result, the only difference between $\oo_i$ and $\oo'_i$ is the reversal of loose components of $G$.

  If we have $\oo_S(\mathcal{R'}) = \oo_S(\mathcal{R})^r$ the same arguments can be used to show that $\oo'_i$ is obtainable from $\oo^r_i$ through reversal of loose components.
\end{proof}

\subsection{Dependent Orientation of Loose Components}
\label{app:lemma:partOrientation}
\begin{lemma}
  \label{lem:alignedMeansDependent}
  Let $\mathcal{G} = (G_1, \dots, G_k)$ be a sunflower proper interval graph with shared graph $S$ and simultaneous enumeration $(\oo_1, \dots, \oo_k)$.
  Let $C \subseteq G_i$ and $C' \subseteq G_j$ be components oriented at a block $B$ of $S$.
  Then there exist two vertices $u,v \in B$ such that $B(u,C) <_{\oo_i} B(v,C)$ and $B(u,C') <_{\oo_j} B(v,C')$.
\end{lemma}
\begin{proof}
  Let $s,t \in B$ be from the ``leftmost'' and ``rightmost'' block of $\oo_i$
  that contain vertices of $B$, respectively, i.e. for all $x \in B $ we have $B(s,C) \leq_{\oo_i} B(x,C) \leq_{\oo_i} B(t,C)$.
  Let $s',t' \in B$ be analogous vertices for $\oo_j$ and $C'$.
  Since $C$ and $C'$ are oriented at $B$ it follows that $B(s,C) \neq B(t,C)$ and $B(s',C') \neq B(t',C')$.
  This means that two vertices $u \in B(s, C) \cap B(s', C')$ and $v \in B(t, C) \cap B(t', C')$ fulfill the lemma.
  We now show that $u$ and $v$ exist, i.e. $B(s, C) \cap B(s', C') \neq \emptyset$ and $B(t, C) \cap B(t', C') \neq \emptyset$.

  We first show that $B(s,C) \cap B(s',C')$ contains $s$ or $s'$.
  Assuming $s' \notin B(s,C)$ it follows that $B(s,C) <_{\oo_i} B(s',C)$ and by definition of simultaneous enumerations $B(s,C') \leq_{\oo_j} B(s',C')$, which implies $B(s,C') = B(s',C')$.
  We use analogous arguments to show that $B(t,C) \cap B(t',C')$ contains $t$ or $t'$.
  Assuming $t' \notin B(t,C)$ it follows that $B(t',C) <_{\oo_i} B(t,C)$ and $B(t',C') \leq_{\oo_j} B(t,C')$, which implies $B(t',C') = B(t,C')$
\end{proof}

\subsection{Proof of Theorem~\ref{the:simEnumChar}}
\simulEnumOperationsTheorem*
\begin{proof}
  We first show that $\rho'$ is a simultaneous enumeration if it is obtained from $\rho$ through reversal of an independent component or a reversible part or if $\rho'=\rho^r$.
  Let  $i,j \in \{1, \dots, k\}$ and $u,v \in V_S$ with $B(u,G_i)\ne B(v,G_i)$. Since $\rho$ is a simultaneous enumeration, we have $B(u,G_i)<_{\sigma_i} B(v,G_i)\Rightarrow B(u,G_j)\le_{\sigma_j} B(v,G_j)$ and $B(u,G_i)>_{\sigma_i} B(v,G_j)\Rightarrow B(u,G_j)\ge_{\sigma_j} B(v,G_j)$.
  Let $\rho'=\rho^r$. Then we have $B(u,G_i)<_{\sigma'_i} B(v,G_j)\Rightarrow B(u,G_i)>_{\sigma_i} B(v,G_j)\Rightarrow B(u,G_j)\ge_{\sigma_j} B(v,G_j) \Rightarrow B(u,G_j)\le_{\sigma'_j} B(v,G_j)$.
  Next let $\rho'$ be obtained by a reversal of an independent component $C$ in $G_i$. Then we have for $x,y\in V_S\cap V(C)$ that $B(x,G_i)=B(y,G_i)$. For $u,v\in V_S$ with $\{u,v\}\not\subseteq V(C)$ we have for $j\in\{1,\dots,k\}$ that $B(u,G_j)\le_{\rho_j} B(v,G_j)\Leftrightarrow (u,G_j)\le_{\rho'_j} B(v,G_j)$ and thus $\rho'$ is also a simultaneous enumeration.
  Finally let $\rho'$ be obtained by reversal of a reversible part $\mathcal C(B)$.
  For $u,v\in V_S $ with $\{u,v\}\not\in V(B)$ we have as in the previous case for $j\in\{1,\dots,k\}$ that $B(u,G_j)\le_{\rho_j} B(v,G_j)\Leftrightarrow (u,G_j)\le_{\rho'_j} B(v,G_j)$.
  For $u,v\in V(B)$ we have for $i\in\{1,\dots,k\}$ with $B(v,G_i)\ne B(u,G_i)$ that $B(u,G_i)<_{\sigma_i} B(v,G_j)\Leftrightarrow B(u,G_i)>_{\sigma'_i} B(v,G_i)$.
  For $i,j\in\{1,\dots,k\}$ with $B(v,G_i)\ne B(u,G_i)$ and $B(v,G_j)\ne B(u,G_j)$, we obtain $B(u,G_i)<_{\sigma'_i} B(v,G_i)\Rightarrow B(u,G_i)>_{\sigma_i} B(v,G_i)\Rightarrow B(u,G_j)\le_{\sigma_j} B(v,G_j) \Rightarrow B(u,G_j)>_{\sigma'_j} B(v,G_j)$. We can conclude tuple $\rho'$ is also a simultaneous enumeration.

  It remains to show that every simultaneous enumeration $\rho'$ can actually be obtained from $\rho$ or $\rho^r$ through the provided reversals.
  By Lemma~\ref{lem:SPIR_seOnly1blockRev} we obtain $\rho'$ from $\rho$ or $\rho^r$ by reversal of loose components. Without loss of generality assume $\rho'$ can be obtained from $\rho$ by reversal of loose components.
  First note that the order of two vertices $u,v\in V_S$ by $\sigma_i$ with $i\in\{1,\dots,k\}$ and $B(u,G_i)\ne B(v,G_i)$ is affected by the reversal of the component $C$ in $G_i$ containing $u$ and $v$ and by no other reversal.
  Assume we have two components $C,C'$ of $G_i,G_j$ oriented at a block $B$ of $S$ where $C$ is reversed and $C'$ is not. By Lemma~\ref{lem:alignedMeansDependent} we have two vertices $u,v\in B$ such that $B(u,C) <_{\sigma_i} B(v,C)$ and $B(u,C') <_{\sigma_j} B(v,C')$. Since $C$ is reversed and $C'$ is not, we obtain $B(u,C) >_{\sigma'_i} B(v,C)$ and $B(u,C') <_{\sigma'_j} B(v,C')$ which contradicts $\rho'$ being a simultaneous enumeration.
  This implies, that for every block $B$ of $S$ either all components oriented at $B$ or none of them are reversed.
  If one of them is not loose, this implies they are all not contained in a reversible part. If they are all loose, then the reversal of all of them is just the reversal of the reversible part at $B$.
  Hence, we actually only reversed independent components and reversible parts.
\end{proof}

\section{Formal Proofs for Sunflower Unit Interval Graphs}
\subsection{Proof of Lemma~\ref{lem:Hsandwitch}}
\Sandwich*
\begin{proof}
  Given a simultaneous unit interval representation $R$ of $\mathcal G$ that realizes $\zeta$, one obtains $H$ as the intersection graph of all unit intervals in $R$ with a fine enumeration $\sigma$ compatible to $\zeta$ and thus with $\alpha\subseteq \sigma$.
  On the other hand, a unit interval representation of $H$ corresponding to a fine enumeration $\sigma\supseteq\alpha$ induces unit interval representations for $G_1,\dots,G_k$ that correspond to $\zeta_1,\dots,\zeta_k$ where $S$ is represented in the same way. 
\end{proof}

\subsection{Implications from the Forbidden Configuration}
\fourVertexCondition*
\begin{proof}
  The condition for three vertices provided of Theorem~\ref{the:3vertexCondition} (\emph{3-vertex condition}) consists of two special cases of the condition for four vertices $v$, $u$, $x$, $w$ of this corollary (\emph{4-vertex condition}) where $v=u$ or $x=w$ while the other three vertices are distinct.  
  On the other hand, if the 4-vertex condition is not met for $v,u,x,w$, then we have either $vx\in E$ or $vx\not\in E$. In the first case the 3-vertex condition is violated by $v,u,x$ and in the second case it is violated by $v,x,w$. Hence, the 3-vertex condition and the 4-vertex condition are equivalent. By Theorem~\ref{the:3vertexCondition} a vertex ordering is a fine enumeration if and only the 4-vertex condition is satisfied. 
\end{proof}


\begin{restatable}{lemma}{HTExtensionStep}
  \label{lem:refHTExtensionStep}
  Let graph $H=(V',E')$ have fine enumeration $\eta$ and let $u,v,w,x\in V'$ with $u\le_\eta v<_\eta w \le_\eta x$.
  Then we have  ~~
  \begin{inparaenum}[(i)]
  \item \label{itm:HTExtensionStep_1} $ux\in E'\Rightarrow vw\in E'$ ~~and~~
  \item \label{itm:HTExtensionStep_2} $vw\not\in E'\Rightarrow ux\not\in E'$.
  \end{inparaenum}
\end{restatable}
\begin{proof}
  Assume one of the implications is false. Then we have $ux\in E(H)$ and $vw\not\in E(H)$.
  But since $u,x$ are adjacent, the fine enumeration property implies $v,x$ are adjacent and thus also $v,w$ are adjacent. A contradiction.
\end{proof}

\begin{restatable}{lemma}{EnumExtensionStep}
  \label{lem:refEnumExtensionStep}
  Let graph $H=(V',E')$ have fine enumeration $\eta$ and let $u,v,w,x\in V'$ with $v<_\eta w$ and $u <_\eta x$ such that $vw\in E',ux\not\in E'$.

  Then we have  ~~
  \begin{inparaenum}[(i)]
  \item \label{itm:EnumExtensionStep_1} $v\le_\eta u\Rightarrow w<_\eta x$ ~~and~~
  \item \label{itm:EnumExtensionStep_2} $x\le_\eta w\Rightarrow u<_\eta v$.
  \end{inparaenum}
\end{restatable}
\begin{proof}
  Assume one of the implications is false. Then we have $v\le_\eta u$ and $x\le_\eta w$. In particular, we obtain the order $v\le_\eta u <_\eta x \le_\eta w$.
  Since $v,w$ are adjacent, this implies with Lemma~\ref{lem:refHTExtensionStep}\eqref{itm:HTExtensionStep_1} that $u,x$ are adjacent. A contradiction.
\end{proof}

\subsection{Scouting}
We describe the information for each step during scouting with the following structure.
An \emph{$m$-scout} of $\zeta$ is a tuple $(\sigma,X)$ where $\sigma\supseteq \alpha$ is a partial order on $V$ and where $X\subseteq V$ with $|X|=m$ such that 
\begin{enumerate}[(S1)]
  \item \label{itm:pzip3} For $w\in X$ and $v\in V\setminus X$ we have $v \not>_\sigma w$.
  \item \label{itm:pzip4} 
    For $(u,v)\in\sigma\setminus\alpha^\star$ with $u\in V_i$ and $v\in V_j$ there is a vertex $s\in V(S)$ such that $u\le_\sigma s\le_\sigma v$ or there are a $(v,s)$-chain $(v=c_l,\dots,c_1=s)$ in $(G_j,\zeta_j)$ and a $(u,s)$-bar $(u=b_l,\dots,b_1=s)$ in $(G_i,\zeta_i)$ with $c_{t},b_{t}\in X$ and $b_{t}\le_\sigma c_{t}$  for $1\le t < l$ with $l\ge 2$.
  \item \label{itm:pzip1} For $v,w\in V_j,u,x\in V_i$ with $w,x\in X$ as well as $v<_\sigma w$ and $u<_\sigma x$ such that $vw\in E_j$ and $ux\not\in E_i$, we have $x\le_\sigma w\Rightarrow u<_\sigma v$. 
\end{enumerate}

Property\,\ref{itm:pzip3} means that we actually go from the right to the left.
Property\,\ref{itm:pzip4} ensures that we have pairs of chains and bars for every newly ordered pair of vertices.
And Property\,\ref{itm:pzip1} ensures $\sigma$ satisfies the left-closed Property~\eqref{eq:left-closed} for $w,x\in X$.

With the following lemma we can grow a $|V|$-scout to obtain a left-closed partial order on $V$.

\begin{restatable}[$\star$]{lemma}{Scouting}
	\label{lem:props2partialorderStep}
	Let $(\sigma,X)$ be an $m$-scout of $\zeta$ with $m<|V|$.
	Then there is an $(m+1)$-scout of $\zeta$ unless $\mathcal G$ has a conflict for $\zeta$.
\end{restatable}
\begin{proof}
  We assume there is no conflict and construct an $(m+1)$-scout $(\tau,X')$ as follows.
  Since $|X|<|V|$ and $\sigma$ is a partial order on $V$, there is a maximal element $x$ in $V\setminus X$.
  We define $X'=X\cup\{x\}$ and obtain $|X'|=m+1$.  
  We first consider $\sigma$ restricted on every pair of graphs $G_i,G_j$ and enhance it according to Property\,\eqref{itm:pzip1} for $X'$ to a partial order $\tau_{i,j}$ that satisfies all properties of a scout.
  Then we take the union of our obtained scouts as $\tau$ and show it already is transitive. The remaining properties of scouts are then easily derived from the partial orders $\tau_{i,j}$ with $1\le i,j\le k$.
  For every pair of graphs $G_i,G_j$ we only need to apply the implication of Property~\eqref{itm:pzip1} once if applied at a specific configuration.
  For $x\in V_i$, we set $u_i$ to be the last vertex in $V_i$ before $x$, that is not adjacent to $x$, i.e., we set $u_i=\max_\alpha \left(V_i\setminus (X'\cup N_{G_i}(x))\right)$; see Figure~\ref{figA:Rij}.
  We set $w_j$ to be the first vertex in $V_j$ with $x\le_\sigma w_j$, i.e., we set $w_j=\min_\alpha\{w\in X'\cap V_j\mid x\le_\sigma w\}$.
  We further set $v_j$ to be the first vertex in $V_j\setminus X'$ that is adjacent to $w_j$, i.e., we set $v_j=\min_\alpha (N_{G_j}(w_j)\setminus X')$. 
  Note that $u_i,w_j,v_j$ may not exist.
  We define $R_{i,j} = \{(u_i,v_j)\}$ or $R_{i,j}=\emptyset$, if $(u_i,v_j)$ does not exist.
  We set $\tau_{i,j}$ to be the transitive closure of $(V_i\times V_j)\cap(\sigma\cup R_{i,j})$.

  \begin{figure}[tb]
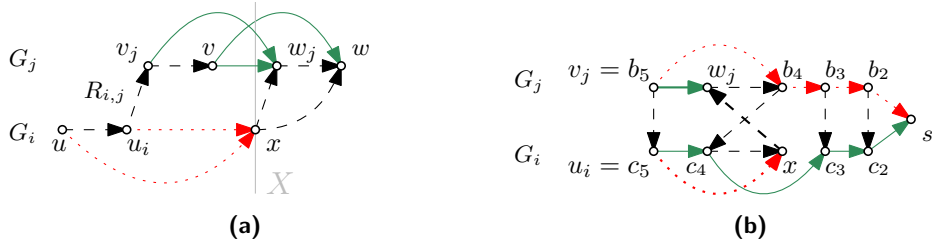
	
    \captionsetup[subfigure]{justification=centering}
    \centering
    \begin{subfigure}[b]{0.45\textwidth}
      \includegraphics{figures/Rij.pdf}
      \caption{}
      \label{figA:Rij}
    \end{subfigure}
    ~
    \begin{subfigure}[b]{0.45\textwidth}
      \includegraphics{figures/tauijAsymmetric.pdf}
      \caption{}
      \label{figA:tauijantisymmetric}
    \end{subfigure}
    \caption{On the left: The vertices $u_i,v_j,w_j$ as derived from $x$ and $X'$. We have $R_{i,j}=\{(u_i,v_j)\}$. On the right: The situation for $(v_j,u_i)\in\sigma\setminus\alpha^\star$. We have the $(v_j,s)$-bar $(v_j,b_4,b_3,b_2,s)$ and the $(u_i,s)$-chain $(u_i,c_4,c_3,c_2,s)$. }  
  \end{figure}

  We first show the relations $\tau_{i,j}$ are partial orders. Since they are by definition reflexive and transitive, it remains to show they are antisymmetric.

  \textbf{$\tau_{i,j}$ is antisymmetric:} For $1\le i,j\le k$ it suffices to show $(v_j,u_i)\not\in \sigma$ since $(u_i,v_j)$ is the only tuple added before building the transitive closure and $\sigma$ itself is transitive and antisymmetric. 
  Assume $v_j\le_h u_i$ with $h\in\{i,j\}$. This implies $i=j$ or $v_j\in V(S)$ or $u_i\in V(S)$. With Lemma~\ref{lem:refEnumExtensionStep}\eqref{itm:EnumExtensionStep_1} we obtain $w_j<_{h} x$ in contradiction to the definition of $w_j$.
  Hence, we can assume $(v_j,u_i)\in\sigma\setminus\alpha^\star$. By Property\,\eqref{itm:pzip4} we then have that there is a vertex $s\in V(S)$ such that there are a $(u_i,s)$-chain $(u_i=c_l,\dots,c_1=s)$ and a $(v_j,s)$-bar $(v_j=b_l,\dots,b_1=s)$ in $(G_j,\zeta_j)$ with $b_{l-1}\le_\sigma c_{l-1}$ and $l\ge 2$; see Figure~\ref{figA:tauijantisymmetric}. I.e., we have $v_j<_j b_{l-1}$ and $u_i<_i c_{l-1}$ and $v_jb_{l-1}\in E_j$ and $u_ic_{l-1}\not\in E_i$. With Lemma~\ref{lem:refEnumExtensionStep} we obtain $w_j\le_j b_{l-1}$ and $c_{l-1}\le_i x$. This yields $x<_\sigma w_j \le_j b_{l-1}\le_\sigma c_{l-1}\le_i x$ which contradicts the antisymmetry of $\sigma$. We obtain that $\tau_{i,j}$ is antisymmetric. It is thus a partial order.

  \textbf{$\tau_{i,j}$ satisfies Property\,\eqref{itm:pzip3}:}
  By choice of $x$ Property\,\eqref{itm:pzip3} is satisfied for $(\sigma,X')$
  Note that for $(u,v)\in \tau_{i,j}\setminus \sigma$ we have that $(u,v)$ is obtained from $(u_i,v_j)$ by transitivity and thus $u\le_{\tau_{i,j}} u_i <_i x$. I.e., we have $u\not\in X'$.
  Therefore Property\,\eqref{itm:pzip3} is satisfied for $(\tau_{i,j},X')$ on $V_i\cup V_j$.

  \begin{figure}[tb]
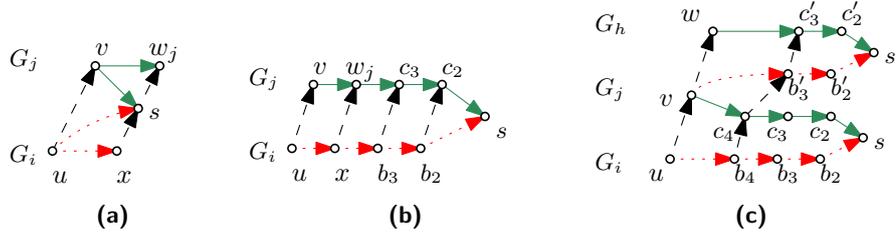
	
    \captionsetup[subfigure]{justification=centering}
    \centering
    \begin{subfigure}[b]{0.2\textwidth} 
      \includegraphics{figures/CBexpansion_left.pdf}
      \caption{}
      \label{figA:CBexpansion_left}
    \end{subfigure}
    ~
    \begin{subfigure}[b]{0.3\textwidth} 
      \includegraphics{figures/CBexpansion_right.pdf}
      \caption{}
      \label{figA:CBexpansion_right}
    \end{subfigure}
    ~
    \begin{subfigure}[b]{0.3\textwidth} 
      \includegraphics{figures/tauTrans.pdf}
      \caption{}
      \label{figA:tauTrans}
    \end{subfigure} 
    \caption{On the left and in the middle: Both cases for the construction of the $(v,s)$-chain and the $(u,s)$-bar to satisfy Property\,\eqref{itm:pzip4} for a tuple $(u,v)\in \tau_{i,j}\setminus\sigma$. On the right: The case of chains and bars in the proof of transitivity for $\tau$.}
  \end{figure}

  \textbf{$\tau_{i,j}$ satisfies Property\,\eqref{itm:pzip4}:}
  Let $u\in V_i,v\in V_j$ with $(u,v)\in \tau_{i,j}\setminus \sigma$; see Figures~\ref{figA:CBexpansion_left},\ref{figA:CBexpansion_right}. Then we have $u\le_\alpha u_i <_\sigma x$ and $v_j \le_\alpha v$. We further have $v<_\sigma w_j$ since otherwise $u<\sigma x \le_\sigma w_j\le_\sigma v$. By Lemma~\ref{lem:refHTExtensionStep} we obtain $vw_j\in E_j$ and $ux\not\in E_i$. By Property\,\eqref{itm:pzip4} there is a vertex $s\in V(S)$ such that $x\le_\sigma s\le_\sigma w_j$ or there are an $(w_j,s)$-chain $(w_j=c_l,\dots,c_1=s)$ in $(G_j,\zeta_j)$ and an $(x,s)$-bar $(x=b_l,\dots,b_1=s)$ in $(G_i,\zeta_i)$ with $c_{t},b_{t}\in X$ and $b_{t}\le_\sigma c_{t}$  for $1\le t < l$ with $l\ge 2$. 
  In the first case we have $s\in X'$ and with Property\,\eqref{itm:pzip3} we obtain $v<_\alpha s\le_\alpha w_j$. We further have $u\le_\alpha u_i<_\alpha x\le_\alpha s$. By Lemma~\ref{lem:refHTExtensionStep} we obtain $vs\in E_j$ and $us\not\in E_i$. In the second case $(v,w_j=c_l,\dots,c_1=s)$ is a $(v,s)$-chain in $(G_j,\zeta_j)$ and $(u,x=b_l,\dots,b_1=s)$ is a $(u,s)$-bar in $(G_i,\zeta_i)$ where we also have $w_j,x\in X'$ and $x\le_\sigma w_j$.
  We therefore have tuple $(v,s)$ is a $(v,s)$-chain of size $2$ in $(G_j,\zeta_j)$ and tuple $(u,s)$ is a $(u,s)$-bar of size $2$ in $(G_i,\zeta_i)$.
  Therefore Property\,\eqref{itm:pzip4} is satisfied for $\tau_{i,j}$.\\
  \textbf{$\tau_{i,j}$ satisfies Property\,\eqref{itm:pzip1}:}
  Let $v,w\in V_j,u,x'\in V_i$ with $x',w\in X'$ as well as $v<_{\tau_{i,j}} w$ and $u<_{\tau_{i,j}} x'$ such that $vw\in E_j$ and $ux'\not\in E_i$. 
  Assume $x\le_{\tau_{i,j}} w$.
  Since $\alpha\subseteq {\tau_{i,j}}$, and relation ${\tau_{i,j}}$ is antisymmetric, and $\alpha$ is a linear order on $V_i$ and on $V_j$, we obtain $v<_\alpha w$ and $u<_\alpha x'$.
  As argued for Property\,\eqref{itm:pzip4}, if $(x',w)\in \tau_{i,j}\setminus \sigma$, then $x'\not\in X'$. Therefore we have $x'\le_\sigma w$.  
  If $x'\ne x$, then we have $x',w\in X$ and thus $u<_\sigma v$ since $(\sigma,X)$ satisfies Property\,\eqref{itm:pzip1}.
  Hence, assume $x'=x$; see Figure~\ref{figA:Rij}.  
  By definition of $v_j,w_j,u_i$ we have $v_j\le_\alpha v<_\alpha w_j\le_\alpha w$ and $u\le_\alpha u_i<_\alpha x$.
  By Lemma~\ref{lem:refHTExtensionStep} we obtain $vw_j\in E_j$ and $ux\not\in E_i$. 
  This yields $u\le_\alpha u_i\le_{R_{i,j}} v_j \le_\alpha v$ and thus $u\le_{\tau_{i,j}} v$. 
  It remains to show $u\ne v$. Assume otherwise. Then we have by antisymmetry of $\tau_{i,j}$ that $u=v$. 
  If $i=j$, then we have $u<_\alpha w,x$ with $uw_j\in E_i$ and $ux\not\in E_i$. This contradicts $x\le_\sigma w_j$ (given by choice of $w_j$) by Lemma~\ref{lem:refEnumExtensionStep} since $\alpha_{|V_i}\subseteq \sigma$ is a fine enumeration of $G_i$. 
  Hence, we have $i\ne j$ and $u=v\in V(S)$. This implies $(x,w_j)\not\in \alpha^\star$.
  By Property\,\eqref{itm:pzip4} we have that there is a vertex $s\in V(S)$ such that $x\le_\sigma s\le_\sigma w_j$ or there are a $(w_j,s)$-chain $(w_j=c_l,\dots,c_1=s)$ in $(G_j,\zeta_j)$ and an $(x,s)$-bar $(x=b_l,\dots,b_1=s)$ in $(G_i,\zeta_i)$ with $c_{t},b_{t}\in X$ and $b_{t}\le_\sigma c_{t}$  for $1\le t < l$ with $l\ge 2$. 
  In the first case we have by Lemma~\ref{lem:refHTExtensionStep} that $us\not\in E_i$ and $vw_j\in E_j$. We then obtain the conflict consisting of the $(u,s)$-chain $(u=v,s)$ in $E_j$ and the $(u,s)$-bar $(u,s)$ in $E_i$ .
  In the second case we obtain the conflict $(C,B)$ with $(u=v,s)$-chain $C=(v,w_j=c_l,\dots,c_1=s)$ in $(G_j,\zeta_j)$ and $(u,s)$-bar $B=(u,x=b_l,\dots,b_1=s)$.
  Since we assumed there is no conflict we can conclude $u\ne v$. \\

  Finally, we define $\tau = \bigcup\limits_{1\le i,j\le k}\tau_{i,j}$.

  \textbf{$\tau$ is transitive:}
  We next show that $\tau$ is already transitive. Note that $\alpha\subseteq \tau_{1,2}\subseteq \tau$.
  Let $u\in V_i,v\in V_j, w\in V_h$ and $u\le_\tau v \le_\tau w$. If $u\le_\alpha v$ or $v\le_\alpha w$, then there is some $\tau_{l,m}$ with $u\le_{\tau_{l,m}} v \le_{\tau_{l,m}} w$ and thus $u\le_\tau w$. 
  Otherwise, we have $(u,v),(v,w)\in\tau\setminus \alpha^\star$ and Property\,\eqref{itm:pzip4} is satisfied for $(u,v),(v,w)$ with regards to $\tau$.
  If there is an $s\in V(S)$ with $u\le_\tau s \le_\tau v$ or $v\le_\tau s \le_\tau w$, then we have $u\le_\tau s \le_\tau w$.
  Otherwise, we have for $(u,v)$ a vertex $s\in V(S)$  with a $(u,s)$-bar $(u=b_l,\dots,b_1=s)$ in $(G_i,\zeta_i)$ and a $(v,s)$-chain $(v=c_l,\dots,c_1=s)$ in $(G_j,\zeta_j)$ such that $b_{l-1}\le_\tau c_{l-1}$ and $b_{l-1},c_{l-1},\in X'$; see Figure~\ref{figA:tauTrans}. 
  Further, we have for $(v,w)$ a vertex $s'\in V(S)$ with a $(v,s)$-bar $(v=b'_t,\dots,b'_1=s')$ in $(G_i,\zeta_i)$ and a $(w,s)$-chain $(w=c'_t,\dots,c'_1=s')$ in $(G_h,\zeta_h)$ such that $b'_{l-1}\le_\tau c'_{t-1}$ and $b'_{t-1}, c'_{t-1}\in X'$.
  This yields $b_{l-1}\le_\tau c_{l-1}$ and $b'_{l-1}\le_\tau c'_{t-1}$ and by the definition of chains and bars we have $ub_{l-1}\not\in\alpha$ and $vc_{l-1}\in \alpha$ and $vb'_{t-1}\not\in\alpha$ and $wc'_{t-1}\in \alpha$.
  By Lemma~\ref{lem:refEnumExtensionStep}\eqref{itm:EnumExtensionStep_1} we have $c_{l-1}\le_j b'_{t-1}$ and thus $b_{l-1}\le_\tau c_{l-1} \le_j b'_{t-1}\le_\tau c'_{t-1}$. 
  This means we have $b_{l-1}\le_\tau c'_{t-1}$ and $u\le_i b_{l-1}$ and $w\le_h c'_{t-1}$ and $ub_{l-1}\not\in \alpha$ and $wc'_{t-1}\in \alpha$.
  By Property\,\eqref{itm:pzip4} for $\tau$ we obtain $u\le_\tau w$. Hence, $\tau$ is transitive. \\
  \textbf{$\tau$ is antisymmetric:}
  Since $\tau$ is transitive it suffices to show there are no two distinct vertices $u,v\in V$ with $(u,v),(v,u)\in\tau$.
  Since $\tau$ is the union of the relations $\tau_{i,j}$ with $1\le i,j \le k$ and those share pairwise at most tuples of a set $V_l$ on which they coincide with $\alpha$, that is the case. We obtain that $\tau$ is a partial order.\\
  \textbf{$\tau$ satisfies Properties\,\ref{itm:pzip3},\ref{itm:pzip4},\ref{itm:pzip1}:}
  Since $\tau$ is the union of the relations $\tau_{i,j}$ with $1\le i,j \le k$ and those satisfy Properties\,\ref{itm:pzip3},\ref{itm:pzip4},\ref{itm:pzip1}, 
  the partial order $\tau$ itself also satisfies Properties\,\ref{itm:pzip3},\ref{itm:pzip4},\ref{itm:pzip1}.
\end{proof}

With the construction of a $|V|$-scout we just obtain a left-closed partial order.

\begin{restatable}[$\star$]{lemma}{Scouted}
	\label{lem:props2partialorder}
  If $\mathcal G$ has no conflict for $\zeta$,
	then there is a left-closed partial order $\sigma\supseteq\alpha$ on $V$.
\end{restatable}
\begin{proof}
  We first show that $(\alpha,X=\emptyset)$ is a $0$-scout of $\zeta$.
  We have that $\alpha$ is a partial order.
  For Property\,\eqref{itm:pzip3} and Property\,\eqref{itm:pzip1} there is nothing to show, since $X=\emptyset$.
  For Property\,\eqref{itm:pzip4} observe that $\alpha\setminus \alpha^\star$ only contains tuples $(u,v)$ obtained by transitivity, which requires some vertex $s\in V(S)$ with $u\le_i s\le_j v$, where $u\in V_i,v\in V_j$.
  By Lemma~\ref{lem:props2partialorderStep} we obtain a $|V|$-scout $(\sigma,V)$ of $\zeta$.
  Let $v,w,u,x\in V$ with $vw\in E$ and $ux\in F$ and $x\le_\sigma w$.
  By definition of $E,F$ we have $1\le i,j\le k$ with $v,w\in V_j,u,x\in V_i$ and $vw\in E_j$ while $ux\not\in E_i$. Hence, we can use Property\,\eqref{itm:pzip1} and obtain $u<_\sigma v$.
\end{proof}

As a result we obtain Lemma~\ref{lem:leftClosedChar}.

\subsection{Zipping}
We use a structure similar to scouts for zipping. 
While scouts demand extensions according to Figure~\ref{fig:EnumExtensionStep_right} we now do so for Figure~\ref{fig:EnumExtensionStep_left} with Property~\ref{itm:zip3} below.
This prevents that our greedy choices violate the property of being left-closed.
At this point the construction always works and we no longer need to consider
chains and bars. Instead we just ensure the partial order remains left-closed.

An \emph{$m$-zip} of $\zeta$ is a tuple $(\sigma,U)$ where $\sigma$ is a partial order on $V$ with $\alpha\subseteq\sigma$
and where $U\subseteq V$ with $|U|=m$ such that 
\begin{enumerate}[(Z1)]
	\item  \label{itm:zip1} For $v\in U$ and $w\in V\setminus U$ we have $v<_\sigma w$.
	\item  \label{itm:zip2} $\sigma$ is a linear order on $U$.
  \item  \label{itm:zip3} For $u,v\in U$ and $w,x\in V$ with $vw\in E$ and $ux\in F$, we have $v\le_\sigma u\Rightarrow w<_\sigma x$. 
  \item  \label{itm:zip4} $\sigma$ is left-closed. 
\end{enumerate}

With the following lemma we grow a zip that spans $V$.
	
\begin{restatable}[$\star$]{lemma}{Zipping}
\label{lem:extendSatifyingPOstep}
	Let $(\sigma,U)$ be an $m$-zip of $\zeta$ with $m<|V|$.
  Then there is an $(m+1)$-zip of $\zeta$.
\end{restatable}
\begin{proof}
  Since $|U|<|V|$ and $\sigma$ is a partial order on $V$, there is a minimal element $u$ in $V\setminus U$.
  We define $U'=U\cup\{u\}$ and obtain $|U'|=m+1$.
  We denote the set of all tuples implied by Property\,\eqref{itm:zip1} by $Q$, i.e., we define $Q=\{u\}\times (V\setminus U)$. 
  We further define the set $R$ to be the set of all tuples implied by $Q$ with Property\,\eqref{itm:zip3},
  i.e., we set $R=\{(w,x)\in (V\setminus U)^2\mid \exists v\in U\colon v\le_\sigma u\wedge vw\in E\wedge ux\in F\}$.
  We finally set $\tau$ to be the transitive closure of $\sigma\cup Q\cup R$.

  By choice of $u$, with no edge in $Q\cup R$ ending in $U$, and with $Q\cup R\subseteq \tau$, the Properties~\eqref{itm:zip1},\eqref{itm:zip2},\eqref{itm:zip3} are satisfied for $\tau$. We prove the following observation for later use: 

  \begin{figure}[tb]	
    \centering
    \captionsetup[subfigure]{justification=centering}
    \begin{subfigure}[b]{0.3\textwidth}
      \includegraphics{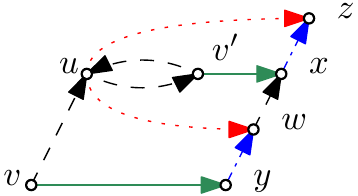}
      \caption{}
      \label{figA:Rtrans}
    \end{subfigure}
    ~
    \begin{subfigure}[b]{0.2\textwidth}
      \includegraphics{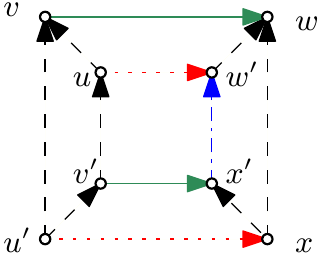}
      \caption{}
      \label{figA:zip4}
    \end{subfigure}
    \caption{On the left: Proof of Observation~\eqref{eq:Rtrans}. On the right: Proof that $\tau$ is left-closed. Arrows for tuples in $R$ are blue and dash dotted.}
  \end{figure}

  \begin{equation}\forall (y,w),(x,z)\in R\colon w\not\le_\sigma x \label{eq:Rtrans}\end{equation}
  Assume there are $(y,w),(x,z)\in R$ with $(w,x)\in\sigma$. See Figure~\ref{figA:Rtrans}.
  By definition of $R$ there are $v,v'\in V$ such that $vy\in E$ and  $uw\in F$ and $v\le_\sigma u$ as well as $v'x\in E$ and $uz\in F$ and $v'\le_\sigma u$.
  Since $\sigma$ is left-closed, we obtain $u<_\sigma v'$, a contradiction.\\

  \textbf{$\tau$ is antisymmetric:} 
  Assume there is a cycle $C$ in graph $G^\star=(V,\sigma\cup Q\cup R)$.
  By Property~\eqref{itm:zip2} and with no edge in $\tau\setminus\sigma$ ending in $U$, we know that $C$ contains no edge of $Q$.
  We can exclude all remaining sequences of edges in $\sigma\cup R$ for $C$ with Observation~\eqref{eq:Rtrans}:
  A cycle with edges in $\sigma$ and in $R$ is excluded with $\sigma$ being transitive. A cycle in $R$ is excluded with $\sigma$ being reflexive. And a cycle in $\sigma$ is excluded by $\sigma$ being antisymmetric. 
  Hence, relation $\tau$ is antisymmetric and thus a partial order.

  \textbf{$\tau$ is left-closed} 
  Let $v,w,u',x\in V$ with $vw\in E$ and $u'x\in F$ and $x\le_\tau w$. 
  If $u'\in U'$ or $v\in U'$, then we have $u'<_\tau v$ or $v<_\tau u'$ by Properties~\eqref{itm:zip1},\eqref{itm:zip2}. The latter case contradicts Property~\eqref{itm:zip3}. I.e., we have $u'<_\tau v$ as desired.

  Otherwise, we have $u',v\in V\setminus U'$ and thus also $x,w\in V\setminus U'$ since $v\le_\alpha w$ and $u'\le_\alpha x$. Then there must be a path in $(V,\sigma\cup R)$ from  $x$ to $w$.
  If $x\le_\sigma w$, then we obtain $u'<_\sigma v$ with Property~\eqref{itm:zip4} of $\sigma$. Otherwise, we obtain with Observation~\eqref{eq:Rtrans} that there are $x',w'\in V$ with $x\le_\sigma x'\le_R w'\le_\sigma w$. See Figure~\ref{figA:zip4}.
  From $(x',w')\in R$ we obtain a vertex $v'\in U$ such that $v'\le_\sigma u$ and $v'x'\in E$. Since $\sigma$ is left-closed, we obtain $u'<_\sigma v'$ and $u<_\sigma v$. This yields $u'<_\sigma v$.
\end{proof}

The result is the following lemma.

\begin{restatable}[$\star$]{lemma}{Zipped}
\label{lem:extendSatisfyingPO}
	Let $\sigma\supseteq\alpha$ be a left-closed partial order on $V$. 
	Then there is a left-closed linear order $\tau\supseteq\alpha$.
 \end{restatable}
\begin{proof}
  Note that $(\sigma,\emptyset)$ is a $0$-zip of $\zeta$.
  We obtain a $|V|$-zip $(\tau,V)$ of $\zeta$ by induction using Lemma~\ref{lem:extendSatifyingPOstep}. The statement holds for $\tau$ by Properties~\eqref{itm:zip2},\eqref{itm:zip4} of a $|V|$-zip.
\end{proof}

As a result we obtain Lemma~\ref{lem:zipping}.

\subsection{Constructing $H$ with fine enumeration $\tau$}
\begin{restatable}[$\star$]{lemma}{Sandwiching}
  \label{lem:extendEF}
  Let $\tau\supseteq\alpha$ be a left-closed linear order on $V$.  
  Then $\tau$ is a fine enumeration for a graph $H=(V,E')$ with $E\subseteq E'$ and $(F\cup F^r)\cap E'=\emptyset$.
\end{restatable}
\begin{proof}
  We set $E'=\{ux\in V^2 \mid \exists vw\in E\colon v\le_\tau u <_\tau x\le_\tau w\}$.
  Clearly, we have $E\subseteq E'$. On the other hand, an edge $ux\in E'\cap F$ would contradict $\tau$ being left-closed and an edge $ux\in E'\cap F^r$ would contradict transitivity of $\tau$. Let $ux\in E'$ with $u\le_\tau x$. Let $y\in V$ with $u <_\tau y <_\tau x$. By definition of $E'$ there are $v,w\in V$ with $v\le_\tau u<_\tau y <_\tau x \le_\tau w$. We obtain $uy,yx\in E'$. Hence, for $v\in V$ the neighborhood $N_H(v)$ is consecutive in $\tau$, and thus $\tau$ is a fine enumeration of $H$. 
\end{proof}

As a result we obtain Lemma~\ref{lem:sandwich}.

\subsection{Recognizing Sunflower Unit Interval Graphs Efficiently}
\label{sub:SUIRalgo}

For our runtime result, we relax the notion of chains and bars.
Let $\mathcal G=(G_1,\dots,G_k)$ be a simultaneous proper interval graph with simultaneous enumeration $\zeta$ and shared graph $S$.
A \emph{relaxed $(u,v)$-chain} is a $(u,v)$-chain in $G^\star=G_1\cup\dots\cup G_k$. A \emph{relaxed conflict at $(u,v)$} is a pair of a relaxed $(u,v)$-chain and a $(u,v)$-bar of the same size.
We relax the result of Theorem~\ref{the:SUIRchar} accordingly in Corollary~\ref{cor:shortestPathsSuffice}.

\begin{corollary}
  \label{cor:shortestPathsSuffice}
  Let $\mathcal G$ be a simultaneous proper interval graph with simultaneous enumeration $\zeta$.
  Then $\mathcal G$ has a simultaneous unit interval representation that realizes $\zeta$ if and only if $\mathcal G$ has no relaxed conflict for $\zeta$.
\end{corollary}
\begin{proof}
  With Theorem~\ref{the:SUIRchar} we obtain the first statement from the last statement since conflicts are also relaxed conflicts. On the other hand, if we have a relaxed $(u,v)$-conflict, then we have a relaxed $(u,v)$-chain and a $(u,v)$-bar both of size $l\ge 2$. As argued before, the relaxed chain implies a distance less than $l-2$ between the intervals $I_u,I_v$ of $u,v$ and the bar implies a distance greater than $l-2$ between $I_u,I_v$ in any simultaneous unit interval representation of $\mathcal G$ that realizes $\zeta$. Thus, such a representation does not exist. Hence, the first statement also implies the second statement.   
\end{proof}

With this preparation we can now start with the actual proof.

\SUIRalg*

\begin{proof} 
  If $\mathcal G$ is not connected (i.e.\,graph $G^\star$ is not connected), then we just combine the  simultaneous unit interval representations of its connected components (i.e.\,the sunflower graphs corresponding to the components of $G^\star$) with some space between them in linear time. If one of the components is not a simultaneous unit interval graph, then $\mathcal G$ is neither. 
  Hence, assume $\mathcal G=(G_1,\dots,G_k)$ is connected. We have $|V|\in O(|E|)$.
  By Theorem~\ref{the:SUIRchar}, sunflower graph $\mathcal G$ is a simultaneous unit interval graph if and only if there is a simultaneous enumeration $\eta$ for which $\mathcal G$ has no conflict. Then also $\eta^r$ has no conflict.
  With Theorem~\ref{the:simEnumChar} we obtain that $\eta$ or $\eta^r$ is obtained from $\zeta$ by reversals of reversible parts and independent components. Hence, we only need to consider such simultaneous enumerations.

  Since every single graph $G_i$ is proper, we only need to consider $(u,v)$-conflicts with $u,v\in V(S)$.
  As argued for the construction of the partial order $\alpha$, we can identify vertices that are in the same block in all graphs $G_1,\dots,G_k$ in which they are both contained. The removed vertices obtain a reference to the corresponding vertex and copy its interval in the end. This is possible in $O(k|V|)\subseteq O(|V|^2)$ time.
  The minimal $(u,v)$-chains for a graph $G$ are exactly the shortest paths in $G$. The size of minimal relaxed $(u,v)$-chains is therefore a lower bound on the size of a minimal $(u,v)$-chain in any graph $G_i$.
  Since relaxed chains are shortest paths in $G^\star$, we obtain all sizes of minimal relaxed chains by breadth-first-searches in $G^\star$ starting at each vertex $v\in V_i=V(G_i)$ with a total runtime in $O(|V|\cdot|E|)$.
  On the other hand, for the maximal size of $(u,v)$-bars in $G_i$ only the reversals of the two corresponding components $C,D$ of $u,v$ are relevant, while components in-between always contribute their maximum independent set independently of whether they are reversed.
  We can reduce the needed time to compute any $(u,v)$-bar with $u,v\in V(S)$ with the following preparation. 
  For every component $C$ of some graph $G_i$ determine the size $\alpha(C)$ of its maximum independent set, which can be computed greedily in linear time, i.e., in $O(k|E|)$ in total.
  Further set $\alpha'(C)=\sum_{D\le_i C} \alpha(D)$. This allows us to compute the size of a maximum independent set strictly between any two components $C',D'$ of $G_i$ with $C'\le_i D'$ as $\alpha'(D')-\alpha'(C')-\alpha(D)$.
  We finally need the maximal bar from every vertex $v\in V(S)$ to both ends of its corresponding component $C$ in $G_i$, if $C$ is loose. These bars again can be greedily computed in linear time, i.e., in $O(|V|\cdot|E|)$ in total.

  We can thus compute for $i,j\in\{1,\dots,k\}$, $u,v\in V(S)$ and each of the four combinations of reversal decisions (reverse or do not reverse) for $C,D$, whether they yield a relaxed conflict at $(u,v)$ in constant time.
  By only considering components that contain not only shared vertices, and for those components considering only the leftmost and the rightmost shared vertex with regards to partial order $\alpha$, we use only $O(|V|^2)$ time in total. 

  If we obtain a relaxed conflict for some combination, then we know by Corollary~\ref{cor:shortestPathsSuffice} that this combination appears in no simultaneous enumeration that can be realized by a simultaneous unit interval representation. 
  Otherwise, we know that every simultaneous enumeration derived with the corresponding combination of reversal decisions has no conflict at $(u,v)$ in $G_i,G_j$.
  We obtain a total of $O(|V|^2)$ considered combinations.

  Since only two components are involved, we can formulate a corresponding $2$-SAT formula $\mathcal F$:
  For every independent component (not contained in the shared graph) and every reversible part, we introduce a literal that represents whether it is reversed or not. For every combination of reversal decisions that yield a conflict we add a clause that excludes this combination. Thereby $\mathcal F$ has $O(|V|^2)$ clauses. Satisfiability of $\mathcal F$ can be decided in linear time by a result of Aspvall et al.~\cite{apt-ltatt-79}, i.e., in $O(|V|^2)$. 
  If $\mathcal F$ is not satisfiable, then every simultaneous enumeration yields a relaxed conflict. 
  Otherwise, we obtain a simultaneous enumeration without conflict by just applying reversals accordingly. 
  By Theorem~\ref{the:SUIRchar}, we then have that $\mathcal G$ is a simultaneous unit interval graph and we obtain a simultaneous unit interval representation by construction along its proof as follows.

  Now consider the iterative construction of scouts according to Lemma~\ref{lem:props2partialorderStep}. 
  Instead of constructing a partial order $\sigma$, we construct a directed acyclic graph $G'$ whose transitive closure is $\sigma$.
  As a preparation we count in every vertex the number of incoming and outgoing edges, directed along $\alpha$. If an edge is added to $G'$, we adapt those counts.
  However, if the end of such an edge is chosen as $x$, we reduce the counter of its start. We can thereby choose a maximal vertex in $V\setminus X$ in constant time, by choosing a vertex with outdegree $0$.
  For every vertex $u\in V$, we further keep track of the first vertex $w$ in every graph $G_i$ according to $\alpha$, with $(u,w)\in E(G')$. 
  Since $x\in V_i$ is directly succeeded by the first vertex $x'$ in $X\cap V_i$, we have $(x,w_j)\in E(G')$ or $x'\le_\sigma w_j$. Thus, we find $w_j$ in constant time.
  We find $v_j,u_i$ as the first vertex adjacent to and before $w_j$ and the last non-adjacent vertex before $x$ as predecessor of the first adjacent one. In total we find them all in $O(|E|)$ time.  

  With adding edge $(u_i,v_i)$ to $G'$ the step of constructing the next scout is complete.
  We next construct zips according to Lemma~\ref{lem:extendSatifyingPOstep}.
  This construction is similar to the construction of the scouts. A minimal vertex $u$ can be found analogously to a maximal vertex $x$ before.
  By adding the vertices of $U$ to a list, we do not need to actually add edges of the form $(u,v)$ with $v\in V\setminus U$ to $G'$.
  Instead of adding all vertices of $R$ to $G'$, it also suffices to compute a maximal $v_j\le_\sigma u$ and to add the corresponding edge $(w_j,x_i)$ with maximal $w_j$ and minimal $x_i$ such that $v_j,w_j$ are adjacent and $u,x_i$ are not. This can be done analogously to finding $(u_i,v_j)$ in the construction of scouts.
  We finally follow the proof of Lemma~\ref{lem:extendEF} to decide adjacency between vertices of different graphs $G_i,G_j$. This can be done in linear time by going from left to right along our linear order of $V$ as follows. We keep track of the last vertex $w$ adjacent to all vertices visited so far, including the current vertex $v$. We then set $v$ to be adjacent to $w$ and to all vertices between $v$ and $w$.
  This takes $O(|V|^2)$ time in total. We obtain a fine enumeration, from which a unit interval representation of a graph $H$ that has $G_1,\dots,G_k$ as induced subgraphs can be obtained in linear time. This yields a simultaneous unit interval representation of $\mathcal G$. 
\end{proof}

\section{Non-Sunflower Simultaneous Proper and Unit Interval Graphs}
\label{ch:hardness}

In this section we consider the problems {\sc PropSimRep} and {\sc SimUintRep} without the restriction of sunflower intersection.  We show that, if the number $k$ of graphs is part of the input, then these problems are NP-complete.
Our reductions are similar to those used by Bok and Jedli\v{c}kov\'a~\cite{bok2018note}.

\begin{theorem}
  \label{thm:propsim-hardness}
  Recognizing Simultaneous Proper Interval Graphs is NP-complete.
\end{theorem}

\begin{proof}
  The problem is clearly in NP, as we can guess the ordering of the
  endpoints of the intervals in a simultaneous representation and
  verify (in polynomial time) whether the resulting representation is
  a simultaneous proper interval representation of the input graphs.

  For the NP-hardness, we present a reduction from the NP-hard problem
  {\sc Betweenness}~\cite{journals/siamcomp/Opatrny79} which, given a
  ground set $A$ and a set $\mathcal T \subseteq A \times A \times A$
  of triplets of $A$ asks whether there exists a linear order
  $\sigma$ of $A$ such that for any triple $(a,b,c) \in \mathcal T$,
  we have $a <_{\sigma} b <_{\sigma} c$ or $c <_{\sigma} b <_{\sigma}
  a$.  We call such an ordering $\sigma$ a betweenness ordering.

  \begin{figure}[tb]
    \centering
    \includegraphics[scale=1]{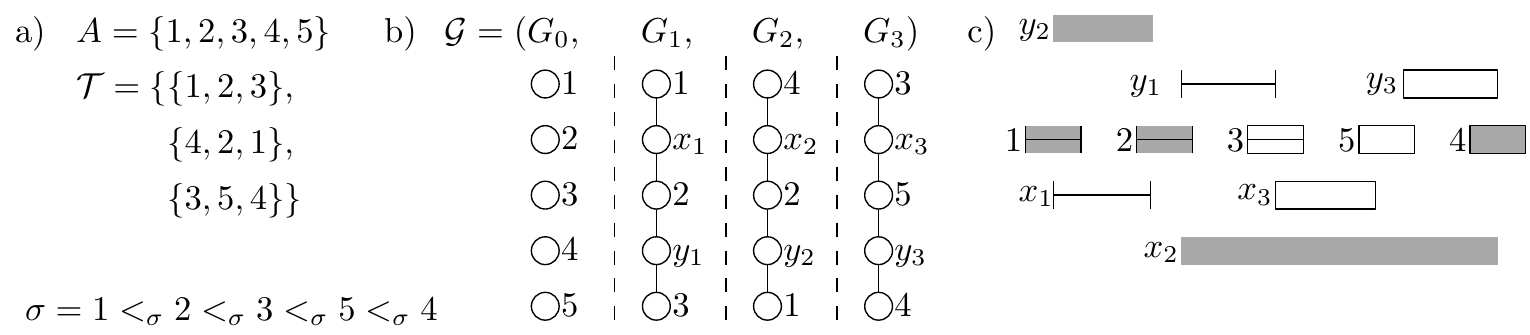}
    \caption{(a) A \textsc{Betweenness} instance $(A, \mathcal{T})$ with solution \oo. (b) The simultaneous graph $\mathcal{G}$ constructed from $(A, \mathcal{T})$. (c) A simultaneous proper interval representation of $\mathcal{G}$. The intervals of $R_1, R_2$, and $R_3$ are represented by lines, gray bordlerless boxes, and bordered boxes, respectively.}
    \label{hardness:fig:betweenIntro}
  \end{figure}

  Let $(A,\mathcal T)$ with $\mathcal T = \{T_1,\dots,T_k\}$ be an
  instance of {\sc Betweenness}.  We construct a simultaneous graph
  consisting of $k+1$ graphs $G_0,\dots,G_k$; see
  Fig.~\ref{hardness:fig:betweenIntro}.  The graph $G_0 =
  (A,\emptyset)$ contains all elements of $A$ as vertices but no
  edges.  For each triple $T_i =(a_i,b_i,c_i)$, we define the graph
  $G_i$ as an induced path $a_ix_ib_iy_ic_i$ where $x_i$ and $y_i$ are
  new vertices.  We set $\mathcal G = (G_0,G_1,\dots,G_k)$ and claim
  that $\mathcal G$ has a simultaneous proper interval representation
  $\mathcal R$ if and only if $(A,\mathcal T)$ admits a betweenness
  ordering $\sigma$.

  If $\mathcal R = (R_0,R_1,\dots,R_k)$ is a simultaneous interval
  representation of $\mathcal G$, then the representation $R_0$
  defines a linear order $\sigma$ of $A$.  The fact that $R_i$ is a
  proper interval representation of an induced path guarantees that
  $b_i$ is positioned between $a_i$ and $c_i$ in $\sigma$ for
  $i=1,\dots,k$.  Therefore $\sigma$ is betweenness ordering for $(A,
  \mathcal T)$.

  Conversely, if $\sigma$ is a betweenness ordering of $(A,\mathcal
  T)$, we use this ordering to define a corresponding representation
  $R_0$ of $G_0$.  For each triple $T_i = (a_i,b_i,c_i)$, due to the
  betweenness property, we can add intervals representing $x_i$ and
  $y_i$ such that we obtain a proper interval representation $R_i$ of
  $R$.  Altogether, this yields a simultaneous proper interval
  representation $\mathcal R = (R_0,R_1,\dots,R_k)$.

  NP-hardness follows since clearly the instance $\mathcal G$ can be
  constructed in polynomial time from $(A,\mathcal T)$.
\end{proof}

\begin{theorem}
  Recognizing Simultaneous Unit Interval Graphs is $NP$-complete.
\end{theorem}

\begin{proof}
  The problem is in $NP$.  Namely, we can guess the ordering of the
  intervals in the representation of each input graph.  Afterwards, a
  unit interval can be described as the solutions of a straightforward
  linear program~\cite{Klavik2017}.

  For the hardness we employ a similar reduction as in the case of
  proper interval graphs in the proof of
  Theorem~\ref{thm:propsim-hardness}.  The key difference is that, the
  vertices in $A$ can easily be represented as unit intervals, the
  vertices $x_i$ and~$y_i$ may span several vertices of $A$, and can
  hence generally not be represented as unit intervals.

  We instead replace $x_i$ and~$y_i$ by a sequence of vertices
  $x_i^1,\dots,x_i^{2n}$ and $y_i^1,\dots,y_i^{2n}$.  For each
  $j=1,\dots,2n-1$, there is a graph $G_i^j$ with
  $V(G_i^j) = \{a_i,b_i,c_i,x_i^j,y_i^j\}$ and edges $x_i^jx_i^{j+1}$
  as well as $y_i^jy_i^{j+1}$.  The edges $a_ix_i^j$, $c_iy_i^j$ are
  present only for $j=1$ and the edges $x_i^{j+1}b_i$, $y_i^{j+1}b_i$
  are present only for $j=2n-1$; see Figure~\ref{fig:hardness-chain}.
  Observe that this construction ensures that the vertices
  $x_i^1,\dots,x_i^{2n}$ all lie between $a_i$ and $b_i$, and likewise
  $y_i^1,\dots,y_i^{2n}$ lie between $b_i$ and $c_i$.  The graph
  $G_i^{2n-1}$ further ensures that they lie on different sides of
  $b_i$, i.e., again a simultaneous representation determines a
  betweenness ordering.

  \begin{figure}[tb]
    \centering
    \includegraphics{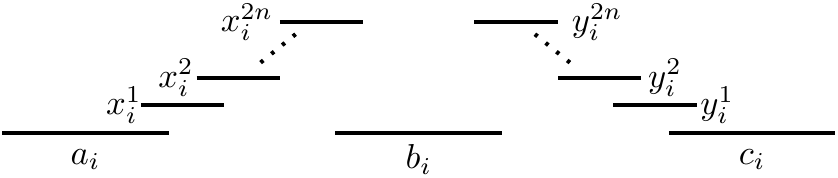}
    \caption{Illustration of the hardness proof for recognition of simultaneous unit interval graphs.}
    \label{fig:hardness-chain}
  \end{figure}
  Moreover, the vertices $x_i^1,\dots,x_i^{2n}$ can be put arbitrarily
  close together or stretched to cover any distance less than $2n$,
  since the only requirement is that consecutive vertices intersect
  each other.  Thus for any betweenness ordering of the vertices in
  $A$ one can construct a corresponding simultaneous unit interval
  representation of the graphs $G_i$, $i=0,\dots,k$ and $G_i^j$ for
  $i=1,\dots,k$, $j=1,\dots,2n-1$.
\end{proof}

\end{document}